\documentclass[aps,pra,onecolumn,letterpaper,superscriptaddress,floatfix,nofootinbib]{revtex4}

\usepackage[utf8]{inputenc}
\usepackage{bm}
\usepackage{times}
 \usepackage{algorithm,algorithmic}

\usepackage{mathtools}
\usepackage{amsmath}
\usepackage[shortlabels]{enumitem}
\usepackage{graphicx,epic,eepic,epsfig,amsmath,latexsym,amssymb,verbatim,color}
 
\usepackage{amsfonts}       % blackboard math symbols
\usepackage{nicefrac}       % compact symbols for 1/2, etc.

\usepackage{amsmath}
\usepackage{bbm}

\usepackage{float}
\usepackage{tikz}
\usetikzlibrary{chains}
\usetikzlibrary{fit}
\usepackage{pgflibraryarrows}		%optional
\usepackage{pgflibrarysnakes}		%optional

\usepackage{epsfig}
\usetikzlibrary{shapes.symbols,patterns} % for source symbols
\usepackage{pgfplots}

\usepackage[strict]{changepage}
\usepackage{hyperref}
\hypersetup{colorlinks=true,citecolor=blue,linkcolor=blue,filecolor=blue,urlcolor=blue,breaklinks=true}

\usepackage[marginal]{footmisc}
\usepackage{url}
\usepackage{theorem}

\newtheorem{definition}{Definition}
\newtheorem{proposition}{Proposition}
\newtheorem{lemma}[proposition]{Lemma}

\def\squareforqed{\hbox{\rlap{$\sqcap$}$\sqcup$}}
\def\qed{\ifmmode\squareforqed\else{\unskip\nobreak\hfil
\penalty50\hskip1em\null\nobreak\hfil\squareforqed
\parfillskip=0pt\finalhyphendemerits=0\endgraf}\fi}
\def\endenv{\ifmmode\;\else{\unskip\nobreak\hfil
\penalty50\hskip1em\null\nobreak\hfil\;
\parfillskip=0pt\finalhyphendemerits=0\endgraf}\fi}
\newenvironment{proof}{\noindent \textbf{{Proof~} }}{\hfill $\blacksquare$}

\newcounter{remark}
\newenvironment{remark}[1][]{\refstepcounter{remark}\par\medskip\noindent%
\textbf{Remark~\theremark #1} }{\medskip}

\newcounter{example}

\mathchardef\ordinarycolon\mathcode`\:
\mathcode`\:=\string"8000
\def\vcentcolon{\mathrel{\mathop\ordinarycolon}}
\begingroup \catcode`\:=\active
  \lowercase{\endgroup
  \let :\vcentcolon
  }

\usepackage{cleveref}
\usepackage{graphicx}
\usepackage{xcolor}

\RequirePackage[framemethod=default]{mdframed}
\newmdenv[skipabove=7pt,
skipbelow=7pt,
backgroundcolor=darkblue!15,
innerleftmargin=5pt,
innerrightmargin=5pt,
innertopmargin=5pt,
leftmargin=0cm,
rightmargin=0cm,
innerbottommargin=5pt,
linewidth=1pt]{tBox}

\newmdenv[skipabove=7pt,
skipbelow=7pt,
backgroundcolor=blue2!25,
innerleftmargin=5pt,
innerrightmargin=5pt,
innertopmargin=5pt,
leftmargin=0cm,
rightmargin=0cm,
innerbottommargin=5pt,
linewidth=1pt]{dBox}
\newmdenv[skipabove=7pt,
skipbelow=7pt,
backgroundcolor=darkkblue!15,
innerleftmargin=5pt,
innerrightmargin=5pt,
innertopmargin=5pt,
leftmargin=0cm,
rightmargin=0cm,
innerbottommargin=5pt,
linewidth=1pt]{sBox}
\definecolor{darkblue}{RGB}{0,76,156}
\definecolor{darkkblue}{RGB}{0,0,153}
\definecolor{blue2}{RGB}{102,178,255}
\definecolor{darkred}{RGB}{195,0,0}
\newcommand{\nc}{\newcommand}
\nc{\rnc}{\renewcommand}
\nc{\beg}{\begin{equation}}
\nc{\eeq}{{\end{equation}}}
\nc{\beqa}{\begin{eqnarray}}
\nc{\eeqa}{\end{eqnarray}}
\nc{\lbar}[1]{\overline{#1}}
\nc{\bra}[1]{\langle#1|}
\nc{\ket}[1]{|#1\rangle}
\nc{\ketbra}[2]{|#1\rangle\!\langle#2|}
\nc{\braket}[2]{\langle#1|#2\rangle}

\nc{\proj}[1]{| #1\rangle\!\langle #1 |}
\nc{\avg}[1]{\langle#1\rangle}
\nc{\rank}{\operatorname{Rank}}
\nc{\smfrac}[2]{\mbox{$\frac{#1}{#2}$}}
\nc{\tr}{\operatorname{Tr}}
\nc{\ox}{\otimes}
\nc{\dg}{\dagger}
\nc{\dn}{\downarrow}
\nc{\cA}{{\cal A}}
\nc{\cB}{{\cal B}}
\nc{\cC}{{\cal C}}
\nc{\cD}{{\cal D}}
\nc{\cE}{{\cal E}}
\nc{\cF}{{\cal F}}
\nc{\cG}{{\cal G}}
\nc{\cH}{{\cal H}}
\nc{\cI}{{\cal I}}
\nc{\cJ}{{\cal J}}
\nc{\cK}{{\cal K}}
\nc{\cL}{{\cal L}}
\nc{\cM}{{\cal M}}
\nc{\cN}{{\cal N}}
\nc{\cO}{{\cal O}}
\nc{\cP}{{\cal P}}
\nc{\cQ}{{\cal Q}}
\nc{\cR}{{\cal R}}
\nc{\cS}{{\cal S}}
\nc{\cT}{{\cal T}}
\nc{\cV}{{\cal V}}
\nc{\cX}{{\cal X}}
\nc{\cY}{{\cal Y}}
\nc{\cZ}{{\cal Z}}
\nc{\cW}{{\cal W}}
\nc{\csupp}{{\operatorname{csupp}}}
\nc{\qsupp}{{\operatorname{qsupp}}}
\nc{\var}{{\operatorname{var}}}
\nc{\rar}{\rightarrow}
\nc{\lrar}{\longrightarrow}
\nc{\polylog}{{\operatorname{polylog}}}
\nc{\wt}{{\operatorname{wt}}}
\nc{\av}[1]{{\left\langle {#1} \right\rangle}}
\nc{\supp}{{\operatorname{supp}}}

\nc{\argmin}{{\operatorname{argmin}}}

\def\x{\xi}

\nc{\RR}{{{\mathbb R}}}
\nc{\CC}{{{\mathbb C}}}
\nc{\FF}{{{\mathbb F}}}
\nc{\NN}{{{\mathbb N}}}
\nc{\ZZ}{{{\mathbb Z}}}
\nc{\PP}{{{\mathbb P}}}
\nc{\QQ}{{{\mathbb Q}}}
\nc{\UU}{{{\mathbb U}}}
\nc{\EE}{{{\mathbb E}}}
\nc{\id}{{\operatorname{id}}}

\nc{\CHSH}{{\operatorname{CHSH}}}

\nc{\be}{\begin{equation}}
\nc{\ee}{{\end{equation}}}
\nc{\bea}{\begin{eqnarray}}
\nc{\eea}{\end{eqnarray}}
\nc{\<}{\langle}
\rnc{\>}{\rangle}
\nc{\rU}{\mbox{U}}

\nc{\ob}[1]{#1}

\nc{\SEP}{{\text{\rm SEP}}}
\nc{\NS}{{\text{\rm NS}}}
\nc{\LOCC}{{\text{\rm LOCC}}}
\nc{\PPT}{{\text{\rm PPT}}}
\nc{\EXT}{{\text{\rm EXT}}}
\nc{\Sym}{{\operatorname{Sym}}}

\nc{\ERLO}{{E_{\text{r,LO}}}}
\nc{\ERLOCC}{{E_{\text{r,LOCC}}}}
\nc{\ERPPT}{{E_{\text{r,PPT}}}}
\nc{\ERLOCCinfty}{{E^{\infty}_{\text{r,LOCC}}}}
\nc{\Aram}{{\operatorname{\sf A}}}

\usepackage{tikz}
\usepackage{hyperref}
\hypersetup{colorlinks=true,citecolor=blue,linkcolor=blue,filecolor=blue,urlcolor=blue,breaklinks=true}

\makeatletter
\def\grd@save@target#1{%
  \def\grd@target{#1}}
\def\grd@save@start#1{%
  \def\grd@start{#1}}
\tikzset{
  grid with coordinates/.style={
    to path={%
      \pgfextra{%
        \edef\grd@@target{(\tikztotarget)}%
        \tikz@scan@one@point\grd@save@target\grd@@target\relax
        \edef\grd@@start{(\tikztostart)}%
        \tikz@scan@one@point\grd@save@start\grd@@start\relax
        \draw[minor help lines,magenta] (\tikztostart) grid (\tikztotarget);
        \draw[major help lines] (\tikztostart) grid (\tikztotarget);
        \grd@start
        \pgfmathsetmacro{\grd@xa}{\the\pgf@x/1cm}
        \pgfmathsetmacro{\grd@ya}{\the\pgf@y/1cm}
        \grd@target
        \pgfmathsetmacro{\grd@xb}{\the\pgf@x/1cm}
        \pgfmathsetmacro{\grd@yb}{\the\pgf@y/1cm}
        \pgfmathsetmacro{\grd@xc}{\grd@xa + \pgfkeysvalueof{/tikz/grid with coordinates/major step}}
        \pgfmathsetmacro{\grd@yc}{\grd@ya + \pgfkeysvalueof{/tikz/grid with coordinates/major step}}
        \foreach \x in {\grd@xa,\grd@xc,...,\grd@xb}
        \node[anchor=north] at (\x,\grd@ya) {\pgfmathprintnumber{\x}};
        \foreach \y in {\grd@ya,\grd@yc,...,\grd@yb}
        \node[anchor=east] at (\grd@xa,\y) {\pgfmathprintnumber{\y}};
      }
    }
  },
  minor help lines/.style={
    help lines,
    step=\pgfkeysvalueof{/tikz/grid with coordinates/minor step}
  },
  major help lines/.style={
    help lines,
    line width=\pgfkeysvalueof{/tikz/grid with coordinates/major line width},
    step=\pgfkeysvalueof{/tikz/grid with coordinates/major step}
  },
  grid with coordinates/.cd,
  minor step/.initial=.2,
  major step/.initial=1,
  major line width/.initial=2pt,
}
\makeatother

\usepackage{thmtools}
\usepackage{thm-restate}
\usepackage{etoolbox}
\makeatletter
\def\problem@s{}
\newcounter{problems@cnt}

\newcommand{\allproblems}{\problem@s}
\makeatother

\definecolor{colorone}{rgb}{1,0.36,0.03}
\definecolor{colortwo}{rgb}{0.4,0.77,0.17}
\definecolor{colorthree}{rgb}{0.01,0.51,0.93}
\definecolor{colorfour}{rgb}{0.47,0.26,0.58}
\usepackage{tcolorbox}
\usepackage{relsize}
\usepackage{graphicx}
\nc{\st}{\text{subject to} \ }
\nc{\supre}{\text{supremum} \ }
\nc{\sdp}{\text{sdp}}
\nc{\cU}{\mathcal U}
\usepackage[qm]{qcircuit}
\usepackage{array}

\usepackage[titletoc,title]{appendix}
\usepackage{graphicx,epic,eepic,epsfig,amsmath,latexsym,amssymb,verbatim,color}
\usepackage{dsfont}
\usepackage{color}
\usepackage[qm]{qcircuit}
\usepackage{float}
\usepackage{tikz}
\usetikzlibrary{chains}
\usetikzlibrary{fit}
\usepackage{pgflibraryarrows}		%optional
\usepackage{pgflibrarysnakes}		%optional
\usepackage{bm}
\usepackage{subfigure}
\usepackage{multirow}
\usetikzlibrary{shapes.symbols,patterns} % for source symbols
\usepackage{pgfplots}
\usepackage[strict]{changepage}
\usepackage{hyperref}
\hypersetup{colorlinks=true,citecolor=blue,linkcolor=blue,filecolor=blue,urlcolor=blue,breaklinks=true}
\usepackage[marginal]{footmisc}
\usepackage{url}
\usetikzlibrary{shapes,arrows}
\usepackage{hyperref}
\newenvironment{proof-sketch}{\noindent{\it Sketch of Proof.}\hspace*{1em}}{\qed\bigskip}

\allowdisplaybreaks

\begin{document}
\title{A Hybrid Quantum-Classical Hamiltonian Learning Algorithm}
\author{Youle Wang}
\affiliation{Center for Quantum Software and Information,
University of Technology Sydney, NSW 2007, Australia}
\affiliation{Institute for Quantum Computing, Baidu Research, Beijing 100193, China}

\author{Guangxi Li}
\affiliation{Center for Quantum Software and Information,
University of Technology Sydney, NSW 2007, Australia}
\affiliation{Institute for Quantum Computing, Baidu Research, Beijing 100193, China}

\author{Xin Wang}
\affiliation{Institute for Quantum Computing, Baidu Research, Beijing 100193, China}
\begin{abstract}
Hamiltonian learning is central to studying complex many-body physics and the certification of quantum devices and simulators. How to learn the Hamiltonian in general with near-term quantum devices is a challenging problem. In this paper, we develop a hybrid quantum-classical Hamiltonian learning algorithm to tackle this problem. By transforming the Hamiltonian learning problem to an optimization problem using the Jaynes’ principle, we employ a gradient-descent method to give the solution and could reveal the interaction coefficients from the system's Gibbs state measurement results. In particular, the computation of the gradients relies on the Hamiltonian spectrum and the log-partition function. Hence, as the main subroutine, we develop a variational quantum algorithm to extract the Hamiltonian spectrum and utilize convex optimization to output the log-partition function. We also apply the importance sampling technique to circumvent the resource requirements for dealing with large-scale Hamiltonians. As a proof of principle, we demonstrate the effectiveness of our algorithm by conducting numerical experiments for randomly generated Hamiltonians and many-body Hamiltonians of theoretical and practical interest. 

\end{abstract}  
\maketitle

\section{Introduction} 
Hamiltonian learning is an important task in studying quantum physics systems and the experimental realization of quantum computers. For instance, it can predict the quantum system's locality to describe the effective interactions between particles, which plays a crucial role in quantum technology, such as quantum lattice models~\cite{Zanardi2002}, quantum simulation~\cite{SethLl2014}, and adiabatic quantum computation~\cite{Aharonov2008}. Moreover, with recent experimental advances in tools for studying complex interacting quantum systems~\cite{Arute2019}, it is becoming more and more essential to learn the dynamics of complicated physical systems, which can predict the evolution of any initial state governed by the Hamiltonian. Another critical utility is relevant to the verification of quantum devices and simulators towards building fault-tolerant quantum computers~\cite{Ladd2010} since certifying that the engineered Hamiltonian matches the theoretically predicted models will always be an indispensable step in developing high-fidelity quantum gates~\cite{Valenti2019}.

{Hamiltonian of many-body physics is often characterized by some parameters, which describe the interactions between the particles. Technically, a many-body Hamiltonian is composed of polynomially many local Pauli operators, i.e.,
\begin{align}
H=\sum_{\ell=1}^{m}\mu_{\ell}E_{\ell},    \label{localhamiltonian}
\end{align}
where $\bm\mu=(\mu_{1},\ldots,\mu_{m})\in[-1,1]^{m}$, and $\{E_{\ell}\}_{\ell=1}^{m}$ are $n$-qubit local Pauli operators, with $m=O(poly(n))$. Despite the number of these parameters $\bm\mu$ in general scales polynomially in the system's size, it is pretty challenging to learn these parameters. Classically characterizing the system's Hamiltonian via tomography would require resources that exponentially scale in the system's size~\cite{gross2010quantum}. Other than tomography, there are methods~\cite{Granade2012,Wiebe2014,Wiebe2014a,Wiebe2015,Wang2017} that cost polynomially many resources while requiring the ability to simulate the dynamics of the system, which is classically intractable. Moreover, it is difficult to perform quantum simulation as a large amount of low-decoherence and fully-connected qubits are required, which are not available on NISQ devices~\cite{Preskill2018}.}

{The major goal of this paper is to learn the many-body Hamiltonians using a trusted NISQ device.  For this purpose, we exploit the variational quantum algorithms (VQAs) that have been gaining popularity in many areas~\cite{xu2019variational,bravo2019variational,huang2019near,larose2019variational,Cerezoa,peruzzo2014variational,nakanishi2019subspace,Wang2020,Wang2020a,Sharma2020}. VQAs are a class of hybrid quantum-classical algorithms that are expected to be implementable on NISQ devices. The main process is to optimize a certain loss function via parameterized quantum circuits (PQCs). In particular, the loss function depending on parameters of the circuit is evaluated on quantum devices, and then the parameters are updated using gradient-based methods classically. As for Hamiltonian learning, we take advantage of the strategy proposed recently in \cite{Anshu2020}, which allows recovering parameters $\bm\mu$ from the measurement results of a quantum Gibbs state $\rho_{\beta}= e^{-\beta H}/\tr(e^{-\beta H})$, i.e., $e_{\ell}=\tr(\rho_{\beta}E_{\ell})$ for all $\ell=1,\ldots,m$. It has been shown that solving the optimization problem below suffices to complete the Hamiltonian learning task.
\begin{align}
    \bm\mu=\text{argmin}_{\bm\nu}\log Z_{\beta}(\bm\nu)+\beta\sum_{\ell=1}^m\nu_{\ell}e_{\ell}. \label{pre:intro_dual}
\end{align}
Here, $Z_{\beta}(\bm\nu)=\tr(e^{-\beta\sum_{\ell=1}^{m}\nu_{\ell}E_{\ell}})$ denotes the partition function, parameterized by $\bm\nu=(\nu_1,...,\nu_m)\in[-1,1]^{m}$, and $\beta$ denotes the inverse temperature of the system. }

{In this paper, we propose a hybrid quantum-classical algorithm to perform the Hamiltonian learning task, whose aim is to recover the interaction coefficients $\bm\mu$ from the measurement results $\{e_{\ell}\}_{\ell=1}^{m}$. The main idea is to solve the optimization problem in Eq.~\eqref{pre:intro_dual} by a gradient-descent method and compute the corresponding gradients utilizing variational quantum algorithms. The challenge of our approach is to compute the log-partition function $\log{Z}_{\beta}(\bm\nu)$ and its gradient since computing partition function is \#P-hard~\cite{goldberg2008inapproximability,long2010restricted}.
}

{To overcome this challenge, we accordingly develop a method based on the relation between the log-partition function and the system's free energy. In general, suppose the state of the system is $\rho$, then the free energy is given by $F(\rho)=\tr(H\rho)-\beta^{-1}S(\rho)$, where $S(\rho)$ is the von Neumann entropy. The relation states that the global minimum of $F(\rho)$ is proportional to the log-partition function, i.e.,
\begin{align}
    \log\tr(e^{-\beta H})=-\beta \min_{\rho}F(\rho).
\end{align}
To establish the results, our method for minimizing the free energy depends on two critical steps. First, we choose a suitable PQC with enough expressiveness and train it to learn the eigenvectors of the Hamiltonian and output the corresponding eigenvalues. Second, we combine the post-training PQC with the classical methods for convex optimization to find the global minimum of the free energy. Next, we utilize the post-training PQC and the optimizer of the convex optimization to compute the gradients. Furthermore, we theoretically analyze the estimation precision of the gradients. We also show the efficiency of loss evaluation and gradients estimation by the importance sampling technique when the underlying Hamiltonian is large.}

{As the proof of principle, we study the effectiveness of our algorithm for Hamiltonian learning by conducting numerical experiments for randomly generated Hamiltonians and several many-body Hamiltonians. To generate random Hamiltonians, we choose Pauli tensor products $E_{\ell}$ from the set $\{X,Y,Z,I\}^{\otimes n}$ at random, with $n$ ranging from 3 to 5. The target interaction coefficients are chosen via a uniform distribution over $[-1,1]$. The tested many-body Hamiltonians consist of Ising, $XY$-spin, and Heisenberg models, where size also varies from 3 to 5 qubits. For these Hamiltonians, we test our algorithm for different parameters $\beta$ and $\bm\mu$ with different lengths. As a result, the numerical results show that the target interaction coefficients can be estimated with high precision. In these experiments, our algorithm learns all eigenvalues of Hamiltonians. Moreover, we show the effectiveness by partially learning few smallest eigenvalues of Ising Hamiltonians. In particular, the circuit depth of used PQC could be significantly reduced. Finally, we also generalize the experiments to larger Ising Hamiltonians with 6/7 qubits.}

{Next, we summarize the contribution of this paper and all mentioned results above.
\begin{enumerate}\setlength{\itemsep}{-0.1cm}
\vspace{-0.2cm}
\item 
We propose a hybrid quantum-classical Hamiltonian learning framework based on the fundamental properties of free energy, which mainly consists of the following two subroutines: log-partition function estimation and stochastic variational quantum eigensolver (SVQE).
\item The main subroutine is the log-partition function estimation algorithm, which combines the SVQE with the classical convex optimization to minimize the free energy. 
\item We also propose a feasible scheme for learning the spectrum of the many-body Hamiltonian by integrating variational quantum algorithms with the importance sampling technique.
\item We demonstrate our algorithm's validity by numerical simulations on several random Hamiltonians and many-body Hamiltonians (e.g., Ising model, XY model, and Heisenberg model).
\vspace{-0.2cm}
\end{enumerate}
}

\textbf{Organization.}  The remaining paper proceeds as follows. In Sec.~\ref{sec_problem}, we formally define the problems we studied in this work; In Sec.~\ref{sec_hqhl}, we present the main results, including the Hamiltonian learning algorithm, and its main subroutines log-partition function estimation, stochastic variational quantum eigensolver, and gradient estimation; In Sec.~\ref{sec_numerical_sim}, we describe the experimental settings and provide numerical results to demonstrate the efficacy of our algorithm; Lastly, we conclude the paper in Sec.~\ref{sec_conclusion}. Proofs and more discussions are presented in the Supplementary Material.

\section{Problem Statement}\label{sec_problem}
In this paper, the goal of Hamiltonian learning is to learn the interaction coefficients $\bm\mu$ from the measurement results of a quantum Gibbs state. We assume that the Hamiltonian to be learned $H$ is composed of local Pauli operators $\{E_{\ell}\}_{\ell=1}^m$, and the measurements corresponding to $\{E_{\ell}\}_{\ell=1}^m$ are performed on the Gibbs state $\rho_{\beta}= e^{-\beta H}/\tr(e^{-\beta H})$ at an inverse temperature $\beta$. The measurement results are denoted by $\{e_{\ell}\}_{\ell=1}^m$, given by
\begin{align}
    e_{\ell}=\tr(\rho_{\beta}E_{\ell}),\quad\forall \ell\in[m].
\end{align}

Recently, there are many methods proposed to efficiently obtain measurement results $\{e_{\ell}\}_{\ell=1}^{m}$~\cite{Cotler2020,Bonet-Monroig2019,Huang2020}. We, therefore, assume the measurement results $\{e_{\ell}\}_{\ell=1}^m$ have been given previously and focus on learning interaction coefficients from them. Formally, we define the Hamiltonian learning problem (HLP) as follows:
\begin{definition}[HLP]
Consider a many-body Hamiltonian with a decomposition given in Eq.~\eqref{localhamiltonian}, where $|\mu_{\ell}|\leq1$ for all $\ell=1,...,m$. Suppose we are given measurement results $\{e_{\ell}\}_{\ell=1}^{m}$ of the quantum Gibbs state $\rho_{\beta}$, then the goal is to find an estimate $\widehat{\bm\mu}$ of $\bm\mu$ such that
\begin{align}
\parallel\widehat{\bm\mu}-\bm\mu\parallel_\infty\leq\epsilon,
\end{align}
where $\|\cdot\|_\infty$ means the maximum norm.
\end{definition}

To solve the HLP, we adopt a strategy that is proposed recently in Ref.~\cite{Anshu2020}, which transforms HLP into an optimization problem by using the Jaynes' principle (or maximal entropy principle)~\cite{jaynes1957information}. This strategy is to find a quantum state with the maximal entropy from all states whose measurement results under $\{E_{\ell}\}_{\ell=1}^m$ match $\{e_{\ell}\}_{\ell=1}^m$. 
\begin{align}
    \max_{\rho}~&\quad S(\rho) \label{optimizationproblem}\\
    s.t.~&~\tr(\rho E_{\ell})=e_{\ell},~\forall\ell=1,...,m\nonumber\\
    &~\rho>0,~\tr(\rho)=1.\nonumber
\end{align}
It has been shown in~\cite{jaynes1957information} that the optimal state is of the following form:
\begin{align}
    \sigma=\frac{\exp(-\beta\sum_{\ell=1}^m \mu_{\ell}^*E_{\ell})}{\tr(\exp(-\beta\sum_{\ell=1}^m \mu_{\ell}^*E_{\ell}))}.
\end{align}
Here, state $\sigma$ is a quantum Gibbs state of a Hamiltonian with interaction coefficients $\bm\mu^*=(\mu_1^*,...,\mu_m^*)$. As a result, Ref.~\cite{Anshu2020} shows that coefficients of $\sigma$ is the target interaction coefficients, i.e., $\bm\mu^*=\bm\mu$. Moreover, Ref.~\cite{Anshu2020} also points out an approach for obtaining $\bm\mu^*$ that is to solve the dual optimization problem in Eq.~\eqref{pre:intro_dual}. 

To this end, we develop a gradient-descent method to solve the problem in Eq.~\eqref{pre:intro_dual}. A flowchart for illustration is shown in Figure~\ref{fig:qahl}. Clearly, the main obstacle is to compute the corresponding gradients of the objective function, which involves computing the partition function. Then, we formalize the gradient estimation problem below.
\begin{definition}[Gradient estimation]
Given a Hamiltonian parameterized by coefficients $\bm\nu$, i.e., $H(\bm\nu)=\sum_{\ell=1}^{m}\nu_{\ell}E_{\ell}$, let $L(\bm\nu)$ be the objective function
\begin{align}
   L(\bm\nu)= \log Z_{\beta}(\bm\nu)+\beta\sum_{\ell=1}^m\nu_{\ell}e_{\ell},
\end{align}
where $Z_{\beta}(\bm\nu)=\tr(e^{-\beta H(\bm\nu)})$. Then the goal is to estimate the gradient $\nabla L(\bm\nu)$ with respect to $\bm\nu$.
\end{definition}
The following sections are devoted to solving HLP and the Gradient estimation problem.
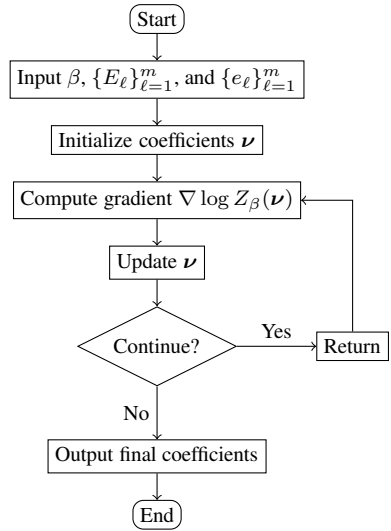
\begin{figure}[t]
    \centering
    \footnotesize
\begin{tikzpicture}[node distance=10pt]
  \node[draw, rounded corners]                        (start)   {Start};
  \node[draw, below=of start]                         (step 1)  {Input $\beta$, $\{E_{\ell}\}_{\ell=1}^m$, and $\{e_{\ell}\}_{\ell=1}^{m}$};
  \node[draw, below=of step 1]                        (step 4)  {Initialize coefficients $\bm\nu$};
  \node[draw, below=of step 4]                        (step 2)  {Compute gradient $\nabla \log Z_{\beta}(\bm\nu)$};
  \node[draw, below=of step 2]                        (step 10)  {Update $\bm\nu$};
  \node[draw, diamond, aspect=2, below=of step 10]     (choice)  {Continue?};
  \node[draw, right=30pt of choice]                   (step x)  {Return};
  \node[draw, below=20pt of choice]  (end)     {Output final coefficients};
  \node[draw, rounded corners, below=of end]                         (step 3)  {End};
  
  \draw[->] (start)  -- (step 1);
  \draw[->] (step 1)-- (step 4);
  \draw[->] (step 4) -- (step 2);
  \draw[->] (step 2) -- (step 10);
  \draw[->] (step 10) -- (choice);
  \draw[->] (choice) -- node[left]  {No} (end);
  \draw[->] (choice) -- node[above] {Yes}  (step x);
  \draw[->] (step x) -- (step x|-step 2) -> (step 2);
  \draw[->] (end) -- (step 3);
\end{tikzpicture}
\caption{Flowchart of the gradient-descent method for Hamiltonian learning.}
\label{fig:qahl}
\end{figure}

\section{Main results}\label{sec_hqhl}
{This section presents the main results of this paper. Specifically, we first discuss the core idea and outline the framework for computing the log-partition function in Sec.~\ref{sec_lpfe}. In Sec.~\ref{sec_hd}, we provide a variational quantum algorithm for learning the eigenvectors of Hamiltonians. Based on the results in Sec~\ref{sec_lpfe}-\ref{sec_hd}, we then proceed to give the gradient estimation procedure in Sec.~\ref{outteroptimization}. Last, Sec.~\ref{sec_hla} provides the main algorithm, the hybrid quantum-classical Hamiltonian learning algorithm (HQHL).}

\subsection{Log-partition function estimation} \label{sec_lpfe}
Here, we consider computing the log-partition function $\log Z_{\beta}(\bm\nu)$. Motivating our method is the relationship between the log-partition function and free energy. Recall that free energy of the system being state $\rho$ is given by $F(\rho)=\tr(H(\bm\nu)\rho)-\beta^{-1} S(\rho)$, assuming the parameterized Hamiltonian is $H(\bm\nu)=\sum_{\ell=1}^{m}\nu_{\ell}E_{\ell}$. Then the relation states that
\begin{align}
    \log Z_{\beta}(\bm\nu)=-\beta\min_{\rho} F(\rho).\label{eq_relation_pre}
\end{align}
As shown in Eq.~\eqref{eq_relation_pre}, it is natural to minimize the free energy to obtain the value of $\log Z_{\beta}(\bm\nu)$. However, it is infeasible to directly minimize the free energy on NISQ devices since performing entropy estimation with even shallow circuits is difficult \cite{Gheorghiu}. To deal with this issue, we choose an alternate version of Eq.~\eqref{eq_relation_pre}:
\begin{align}
    \log Z_{\beta}(\bm\nu)=-\beta\min_{\mathbf{p}} \sum_{j=1}^{N} p_{j}\cdot \lambda_{j}+\beta^{-1} \sum_{j=1}^{N}p_{j}\log p_{j},\label{eq_objective_pre}
\end{align}
where $\bm\lambda=(\lambda_1,...,\lambda_N)$ is the vector of eigenvalues of $H(\bm\nu)$, and $\mathbf{p}=(p_1,...,p_N)$ represents an $N$-dimensional probability distribution, with $N=2^n$ the Hamiltonian's dimension. Please note that proofs for Eqs.~\eqref{eq_relation_pre}-\eqref{eq_objective_pre} are provided in the supplementary file. Thus, optimizing the R.H.S of Eq.~\eqref{eq_objective_pre} could obtain the desired quantity and avoid the von Neumann entropy estimation simultaneously, assuming eigenvalues of the Hamiltonian $H(\bm\nu)$ is given previously. As a result, our task is reduced to solve the following optimization program based on the equality in Eq.~\eqref{eq_objective_pre}:
\begin{align}
    \min_{\mathbf{p}}&\quad C(\mathbf{p}) \label{alternate_optimizaiton}\\
    \mbox{s.t.} &\quad \sum_{j=1}^{N}p_{j}=1\nonumber\\
    & \quad p_{j}\geq0,\forall j=1,\ldots,N \nonumber
\end{align}
where 
\begin{align}
    C(\mathbf{p})=\sum_{j=1}^{N} p_{j}\cdot \lambda_{j}+\beta^{-1} \sum_{j=1}^{N}p_{j}\log p_{j}. \label{op_partion}
\end{align}

The optimization program in Eq.~\eqref{alternate_optimizaiton} is a typical convex optimization program. In the context of convex optimization, there are many classical algorithms to solve the optimization program, such as the interior-point method~\cite{karmarkar1984new}, ellipsoid method~\cite{grotschel1993geometric}, cutting-plane method~\cite{kelley1960cutting}, and random walks~\cite{Kalai2006}, etc. For example, we consider using the cutting plane method \cite{Lee2015,Jiang2020}, which requires the membership and evaluation procedures \cite{Lee2017}. Concerning the program in Eq.~\eqref{alternate_optimizaiton}, the membership procedure determines whether a point belongs to the set of probability distributions, and the evaluation procedure takes in a probability distribution $\mathbf{p}$ and returns the value $C(\mathbf{p})$ with high accuracy. Clearly, it is easy to determine whether the given point is a probability distribution while challenging to efficiently evaluate the function value. Thus, we provide a procedure to solve the convex optimization problem as well as overcome this challenge at the same time in Algorithm~\ref{alg:lpfe}.

\begin{algorithm}[H]
\footnotesize
\caption{Log-partition function estimation}
\label{alg:lpfe}
\begin{algorithmic}[1]
    \REQUIRE Parameterized quantum circuit $U(\bm\theta)$, Hamiltonian $H(\bm\nu)$, constant $\beta$;
    \ENSURE An estimate for $\log{Z}_{\beta}(\bm\nu)$;
    \STATE {\color{blue}\texttt{\# Evaluation procedure construction}}
    \STATE Take probability distribution $\mathbf{p}$ as input;
    \STATE Set integer $T$ and $D$;
    \STATE Sample $TD$ integers $t_{1}^{1},...,t_{T}^{1},...,t_{1}^{D},...,t_{T}^{D}$ according to $\mathbf{p}$;
    \STATE Prepare computational states $\ket{\psi_{t_{1}^{1}}}$, $...$, $\ket{\psi_{t_{T}^{1}}}$, $...$, $\ket{\psi_{t_{1}^{D}}}$, $\ldots$, $\ket{\psi_{t_{T}^{D}}}$;
    \STATE Compute approximate eigenvalues: $\lambda_{t_j^{s}}=\bra{\psi_{t_j^s}}U^{\dagger}(\bm\theta)H(\bm\nu)U(\bm\theta)\ket{\psi_{t_j^s}}$ for all $j=1,\ldots,T$ and $s=1,\ldots,D$;
    \STATE Compute averages: $ave_{s}=\frac{1}{T}\sum_{j=1}^{T}\lambda_{t_{j}^s}$ for all $s=1,...,D$;
    \STATE Take the median value $C(\mathbf{p})\leftarrow{\rm median}(\lambda_{ave_1},...,\lambda_{ave_{D}})+\beta^{-1}\sum_{j=1}^{N}p_{j}\log{p_{j}}$;
    \STATE {\color{blue}\texttt{\# Membership procedure construction}}
    \STATE Construct a membership procedure;
    \STATE {\color{blue}\texttt{\# Convex optimization solution}}
    \STATE Compute the function's global minimum value $C(\mathbf{p}^*)$ and the optimal point $\mathbf{p}^*$ via the cutting plane method.
    \RETURN{value $-\beta C(\mathbf{p}^*)$ and the final point $\mathbf{p}^*$.}
\end{algorithmic}
\end{algorithm}

In Algorithm~\ref{alg:lpfe}, we compute the log-partition function using a classical convex optimization method. For this purpose, we first show the construction process of evaluation procedure. That is, given a point $\mathbf{p}$, find an estimate for $C(\mathbf{p})$. We assume we are given a parameterized quantum circuit $U(\bm\theta)$ that can learn eigenvectors of the Hamiltonian $H(\bm\nu)$. In our approach, the $U(\bm\theta)$ is combined with the importance sampling technique (cf. lines 3-8) to deal with the large-sized Hamiltonians. Specifically, i) we sample $TD$ indices according to the distribution $\mathbf{p}$ (cf. line 4); ii) we evaluate the eigenvalues associated with the sampled indices (cf. lines 5-6); iii) we take the average over $T$ (cf. line 7) and the median over $D$ (cf. line 8) to evaluate the function value $C(\mathbf{p})$ with high accuracy and success probability. Eventually, with the evaluation procedure and the membership procedure, the global minimum of $C(\mathbf{p})$ could be obtained via the cutting plane method \cite{Lee2015,Lee2017,Jiang2020}. Finally, based on the relationship between $\log Z_{\beta}(\bm\nu)$ and $C(\mathbf{p}^*)$ (cf. Eq. (\ref{eq_objective_pre})), we could derive the log-partition function value. Here $\mathbf{p}^*$ denotes the optimal distribution of the optimization in Eq. (\ref{eq_objective_pre})).

\begin{remark}
Notice that a crucial gadget in Algorithm~\ref{alg:lpfe} is the PQC $U(\bm\theta)$, which we have assumed to be accessible. To complement the assumption, we provide a procedure for extracting eigenvalues in the next section, \emph{Stochastic variational quantum eigensolver}. In particular, we will prsent a variational quantum algorithm for learning the eigenvectors of the parameterized Hamiltonians.
\end{remark}

Now we discuss the cost of applying Algorithm~\ref{alg:lpfe}. As the efficiency of Algorithm~\ref{alg:lpfe} mainly relies on the cost of the evaluation procedure, we only discuss it here. Suppose we have access to Hamiltonian $H(\bm\nu)$'s eigenvalues $\bm\lambda$, then the objective function $C(\mathbf{p})$ can be effectively evaluated. Recall that $C(\mathbf{p})$ contains two parts $\sum_{j=1}^{N} p_{j}\cdot \lambda_{j}$ and $\beta^{-1} \sum_{j=1}^{N}p_{j}\log p_{j}$. On the one hand, the latter value can be computed immediately since $\mathbf{p}$ is stored on classical devices. On the other hand, value $\sum_{j=1}^{N} p_{j}\cdot \lambda_{j}$ can be regarded as an expectation of the probability $\mathbf{p}$, where value $\lambda_j$ is sampled with probability $p_j$. Notably, the total cost for estimating $C(\mathbf{p})$ is dominated by the number of samples. Then we analyze the number of required samples for loss evaluation in Proposition~\ref{theorem_logpartitionfunction_pre}.

\begin{proposition}\label{theorem_logpartitionfunction_pre}
For any constant $\beta>0$ and parameterized Hamiltonian $H(\bm\nu)=\sum_{\ell=1}^{m}\nu_{\ell}E_{\ell}$ with $E_{\ell}\in\{X, Y, Z, I\}^{\otimes n}$ and $\bm\nu\in\mathbb{R}^{m}$, suppose we are given access to a parameterized quantum circuit $U(\bm\theta)$ that can prepare $H(\bm\nu)$'s eigenvectors, then the objective function $C(\mathbf{p})$ can be computed up to precision $\epsilon$ with probability larger than $2/3$ by taking $T=O({m\|\bm\nu\|_2^2}/{\epsilon^2})$ samples. Furthermore, the probability can be improved to $1-\eta$ costing an additional multiplicative factor of $D=O(\log(1/\eta))$.
\end{proposition}
\begin{proof-sketch}
In general, the expectation can be approximated by the sample mean according to Chebyshev's inequality. Specifically speaking, the expectation can be estimated up to precision $\epsilon$ with high probability (e.g., larger than $2/3$) by taking $O({\bf Var}/\epsilon^2)$ samples, where ${\bf Var}$ denotes the variance of the distribution. Here, the number of samples is $T=O(m\|\bm\nu\|_2^2/\epsilon^2)$, since the variance is bounded by the squared spectral norm of $H(\bm\nu)$, which is less than $\sqrt{m}\|\bm\nu\|_2$. Furthermore, Chernoff bounds allow improving success probability to $1-\eta$ at an additional cost of a multiplicative factor of $D=O(\log(1/\eta))$.
\end{proof-sketch}

As shown in Proposition~\ref{theorem_logpartitionfunction_pre}, our evaluation method is computationally efficient, since the number of samples scales polynomially with the number of qubits. Hence Algorithm~\ref{alg:lpfe} could be applied to compute the partition function of the parameterized Hamiltonian, given the suitable PQC $U(\bm\theta)$.

\subsection{Stochastic variational quantum eigensolver}~\label{sec_hd}
This section discusses learning the eigenvectors of the parameterized Hamiltonian $H(\bm\nu)$ using variational quantum algorithms and the importance sampling technique. First, we outline the algorithm in Algorithm~\ref{alg:svqe} and then discuss the fundamental theory. Second, we circumvent the cost for coping with large-scaled Hamiltonians by the importance sampling technique. We also analyze the cost of loss evaluation in the algorithm.
\begin{figure}[thb]
\begin{algorithm}[H]
\footnotesize
\caption{Stochastic variational quantum eigensolver (SVQE)}
\label{alg:svqe}
\begin{algorithmic}[1]
    \REQUIRE Parameterized quantum circuit $U(\bm\theta)$, Hamiltonian $H(\bm\nu)$, and weights $\mathbf{q}$;
    \ENSURE Optimal PQC $U(\bm\theta)$;
    \STATE Set number of iterations $I$ and $l=1$;
\STATE Set integers $T$ and $D$;
\STATE Set learning rate $r_{\theta}$;
\STATE Set probability distribution $\mathbf{q}$;
\STATE Sample $TD$ integers $k_{1}^{1}$,$\ldots$,$k_{T}^{1}$,$\dots$,$k_{1}^{D}$,$\ldots$,$k_{T}^{D}$ according to $\mathbf{q}$;
\STATE Prepare computational states $\ket{\psi_{k_{1}^{1}}}$, $\ldots$ , $\ket{\psi_{k_{T}^{1}}}$, $\ldots$, $\ket{\psi_{k_{1}^{D}}}$, $\ldots$, $\ket{\psi_{k_{T}^{D}}}$;
\WHILE{$l\leq I$}
\STATE Compute value $\bra{\psi_{k_{j}^{s}}}U^{\dagger}(\bm\theta)H(\bm\nu)U(\bm\theta)\ket{\psi_{k_{j}^{s}}}$ for all $j=1$,$\ldots$,$T$ and $s=1$,$\ldots$,$D$;
\STATE Compute averages: $ave_{s}=\frac{1}{T}\sum_{j=1}^{T}\bra{\psi_{k_{j}^{s}}}U^{\dagger}(\bm\theta)H(\bm\nu)U(\bm\theta)\ket{\psi_{k_{j}^{s}}}$ for all $s=1,...,D$;
\STATE Let $M(\bm\theta)\leftarrow {\rm median}(ave_{1},...,ave_{D})$;
\STATE Use $M(\bm\theta)$ to compute the gradient $\nabla$ by parameter shift rules~\cite{mitarai2018quantum};
\STATE Update parameters $\bm\theta\leftarrow \bm\theta-r_{\theta}\nabla$;
\STATE Set $l\leftarrow l+1$;
\ENDWHILE
\RETURN{the final $U(\bm\theta)$.}
\end{algorithmic}
\end{algorithm}
\end{figure}

To incorporate variational quantum algorithms, we utilize the variational principle of Hamiltonian's eigenvalues. That is, Hamiltonian's eigenvalues majorize the diagonal elements, and the dot function with an increasingly ordered vector is Schur concave~\cite{roberts1973convex}. A similar idea has already been discussed in \cite{Nakanishi2018}. In contrast, our method learns the full spectrum of the Hamiltonian. We define a function $M(\bm\theta)$ over all parameters $\bm\theta$ of the circuit.
\begin{align}
    M(\bm\theta)=\sum_{j=1}^{N}q_{j}\cdot\bra{\psi_j}U^{\dagger}(\bm\theta)H(\bm\nu)U(\bm\theta)\ket{\psi_j},
\end{align}
where $\mathbf{q}=(q_1,...,q_N)$ is a probability distribution such that $q_1<q_2<...<q_N$, and notations $\ket{\psi_1},\ldots,\ket{\psi_N}$ denote the computational basis. Suppose that PQC $U(\bm\theta)$ has enough expressiveness, then $U(\bm\theta)\ket{\psi_j}$ could learn the $j$-th eigenvector of the Hamiltonian $H(\bm\nu)$ with suitable parameters. Particularly, $M(\bm\theta)$ will reach the global minimum when all eigenvectors are learned. In other words, we use the PQC $U(\bm\theta)$ to learn eigenvectors via finding the global minimum of $M(\bm\theta)$ over all parameters $\bm\theta$.

\begin{remark}
Choosing a suitable $U(\bm\theta)$ is critical to many variational quantum algorithms as well as our Algorithm~\ref{alg:svqe}. With enough expressibility, training the PQC $U(\bm\theta)$ would allow us to exactly or approximately learn the solution to the certain problem. The expressibility of PQCs has been recently studied in~\cite{Sim2019}. Throughout this paper, we assume the used PQC $U(\bm\theta)$ is able to learn well the eigenvectors of Hamiltonians $H(\bm\nu)$ for arbitrary $\bm\nu$.
\end{remark}

\begin{remark}
In the learning process, we employ a gradient-based method to update the parameters $\bm\theta$ iteratively. In each iteration, the corresponding gradients are computed via the parameter shift rule~\cite{mitarai2018quantum}, which outsources the gradient estimation to the loss evaluation. As this is similar to other variational quantum algorithms, we omit the details of gradient computation. For details of gradient derivation, please refer to the proof of Proposition 3 in \cite{Wang2020a}.
\end{remark}

Notice that for large Hamiltonians, the loss $M(\bm\theta)$ may consist of exponentially many terms, which would be a huge burden to the loss evaluation. However, we could employ the importance sampling technique to circumvent this issue. To this end, $M(\bm\theta)$ is taken as an expectation of the distribution $\mathbf{q}$. Hence, $M(\bm\theta)$ is to be estimated by the sample mean. Notably, the cost of loss evaluation is dominated by the number of samples, which is why we call our method stochastic variational quantum eigensolver (SVQE). Our algorithm with importance sampling for minimizing $M(\bm\theta)$ is depicted in Algorithm~\ref{alg:svqe}. In the following, we analyze the sample complexity in the loss evaluation.

\begin{proposition}\label{theorem_SVQE_pre}
Consider a Hamiltonian $H(\bm\nu)=\sum_{\ell=1}^{m}\nu_{\ell}E_{\ell}$ with Pauli operators $E_{\ell}\in\{X,Y,Z,I\}^{\otimes n}$ and constants $\nu_{\ell}\in[-1,1]$. Given any constants $\epsilon>0$, $\eta\in(0,1)$, $\beta>0$, the objective function $M(\bm\theta)$ in SVQE can be estimated up to precision $\epsilon$ with probability at least $1-\eta$, costing $TD$ samples with $T=O(m\|\bm\nu\|_2^2/\epsilon^2)$ and $D=O(\log(1/\eta))$. Besides, the total number of measurements is given below:
\begin{align}
O\left(\frac{mTD\|\bm\nu\|_1^2(n+\log(m/\eta))}{\epsilon^2}\right).
\end{align}
\end{proposition}
\begin{proof-sketch}
The number of samples is determined by the accuracy $\epsilon$ and Hamiltonian $H(\bm\nu)$. By Chebyshev's inequality, estimating $M(\bm\theta)$ up to precision $\epsilon$ with high probability requires $T=O(m\|\bm\nu\|_2^2/\epsilon^2)$ samples, since the variance is bounded by the spectral norm, which is less than $\sqrt{m}\|\bm\nu\|_2$. Meanwhile, the expectation value $\bra{\psi_j}U^{\dagger}(\bm\theta)H(\bm\nu)U(\bm\theta)\ket{\psi_j}$ is evaluated by measurements. We compute the expectation value of the observable $H(\bm\nu)$ by measuring each Pauli operator $E_{\ell}$ separately, since there are only $m=O(poly(n))$ Pauli operators.
\end{proof-sketch}

\begin{remark}
Other methods for computing expectation value of Hamiltonians can be found in Ref.~\cite{Sweke2019,Arrasmith2020}, where importance sampling is employed to sample Pauli operator $E_{l}$ of the Hamiltonian. 
\end{remark}
\begin{remark}
In the context of quantum algorithms, there are many proposed methods for learning the low-lying eigenvectors of the Hamiltonian and diagonalizing Hamiltonian. Some known quantum algorithms for Hamiltonian diagonalization are based on quantum fast Fourier transform \cite{PhysRevLett.83.5162}, which may be too costly for NISQ computers and thus not suitable for our purpose. Recently, there have already been some works on finding ground and excited eigenstates of the Hamiltonian with NISQ devices, i.e., variational quantum eigensolvers~\cite{peruzzo2014variational,higgott2019variational,mcclean2016theory,Nakanishi2018,jones2019variational,kandala2017hardware,Wang2020a,Commeau2020}. They maybe employed to learn eigenvectors in the Hamiltonian learning framework.
\end{remark}

\subsection{Gradient estimation}\label{outteroptimization}
Recall that we employ a gradient-based method to do the optimization in the Hamiltonian learning (cf. Figure~\ref{fig:qahl}). We use the tools developed in Sec.~\ref{sec_lpfe}-\ref{sec_hd} to derive the gradient estimation procedure.

Usually, with the estimated gradient, parameters are updated in the following way:
\begin{align}
    \bm\nu\leftarrow \bm\nu-r\nabla L(\bm\nu),\label{eq:coefficient_update}
\end{align}
where $r$ is the learning rate. The expression of the gradient is given below.
\begin{align}
    \nabla L(\bm\nu)=\left(\frac{\partial L(\bm\nu)}{\partial \nu_{1}},...,\frac{\partial L(\bm\nu)}{\partial \nu_{m}}\right).
\end{align}
Furthermore, the explicit formula of each partial derivative is given in~\cite{Anshu2020}:
\begin{align}
    \frac{\partial L(\bm\nu)}{\partial\nu_{\ell}}&=\frac{\partial}{\partial\nu_{\ell}}\log{Z}_{\beta}(\bm\nu)+\beta e_{\ell}=-\beta\tr(\rho_{\beta}(\bm\nu)E_{\ell})+\beta e_{\ell},\label{outterlossfunctionoptimization}
\end{align}
where $\rho_{\beta}(\bm\nu)= e^{-\beta H(\bm\nu)}/Z_{\beta}(\bm\nu)$ represents the Gibbs state associated with the parameterized Hamiltonian $H(\bm\nu)$. 

\begin{figure}[htb]
\begin{algorithm}[H]
\footnotesize
\caption{Gradient estimation}
\label{alg:grad}
\begin{algorithmic}[1]
    \REQUIRE Post-training circuit $U(\bm\theta)$, Pauli operators $\{E_{\ell}\}_{\ell=1}^{m}$, optimal $\widehat{\mathbf{p}}^*$, and constants $\beta$ and $\{e_{\ell}\}_{\ell=1}^{m}$;
    \ENSURE Gradient estimate $\nabla L(\bm\nu)$;
    \STATE Set $\ell=1$;
\STATE Set integer $K$ and $D$;
\STATE Sample $K$ integers $l_{1}^{1},...,l_{K}^{1},...,l_{1}^{D},...,l_{K}^{D},$ according to $\widehat{\mathbf{p}}^*$;
\STATE Prepare computational states $\ket{\psi_{l_{1}^{1}}}$,$\ldots$,$\ket{\psi_{l_{K}^{1}}}$,$\ldots$,$\ket{\psi_{l_{1}^{D}}}$,$\ldots$, $\ket{\psi_{l_{K}^{D}}}$;
\WHILE{$\ell\leq m$}
\STATE Compute value $\bra{\psi_{l_{j}^s}}U^{\dagger}(\bm\theta)E_{\ell}U(\bm\theta)\ket{\psi_{l_{j}^s}}$ for $j=1,..,K$ and $s=1,...,D$;
\STATE Calculate averages: $ave_{s}=\frac{1}{K}\sum_{j=1}^{K}\bra{\psi_{l_j^{s}}}U^{\dagger}(\bm\theta)E_{\ell}U(\bm\theta)\ket{\psi_{l_{j}^{s}}}$
for all $s=1,...,D$;
\STATE Take the median value: $s_{\ell}=-\beta\cdot{\rm median}(ave_{1},\ldots,ave_{D})+\beta e_{\ell}$;
\STATE Set $\ell\leftarrow \ell+1$;
\ENDWHILE
\RETURN{vector $(s_{1},...,s_{m})$.}
\end{algorithmic}
\end{algorithm}
\end{figure}

According to the second equality in Eq.~\eqref{outterlossfunctionoptimization}, preparing Gibbs state $\rho_{\beta}(\bm\nu)$ is likely to be necessary to the gradient estimation, which is quite challenging~\cite{islam2015measuring,yuan2019theory,wu2019variational,Xu2019,Wang2020}. However, we provide a procedure for gradient estimation without preparing the Gibbs state $\rho_{\beta}(\bm\nu)$ in Algorithm~\ref{alg:grad}. We use the post-training PQC $U(\bm\theta)$ and the optimal distribution $\widehat{\mathbf{p}}^*$ (cf. Algorithm~\ref{alg:lpfe}) from Sec.~\ref{sec_lpfe}-\ref{sec_hd}, respectively. And the component of the gradient can be computed in the sense that
\begin{align}
    \frac{\partial L(\bm\nu)}{\partial\nu_{\ell}}\approx-\beta\sum_{j=1}^{N}\widehat{p}_{j}^*\cdot \bra{\psi_j}U^{\dagger}(\bm\theta)E_{\ell}U(\bm\theta)\ket{\psi_j}+\beta e_{\ell}.\label{approx_derivative}
\end{align}
The validity of the relation in Eq.~\eqref{approx_derivative} is proved in Proposition~\ref{theoremforgibbsstatepreparation}. 

\begin{proposition}[Correctness]
\label{theoremforgibbsstatepreparation}
Consider a parameterized Hamiltonian $H(\bm\nu)$ and its Gibbs state $\rho_{\beta}(\bm\nu)$. Suppose the $U(\bm\theta)$ from SVQE (cf. Algorithm~\ref{alg:svqe}) and $\widehat{\mathbf{p}}^*$ from log-partition function estimation procedure (cf. Algorithm~\ref{alg:lpfe}) are optimal. Define a density operator $\rho_{\beta}^*$ as follows:
\begin{align}
    \rho_{\beta}^*=\sum_{j=1}^{N}\widehat{p}_{j}^*\cdot U(\bm\theta)\op{\psi_j}{\psi_j}U^{\dagger}(\bm\theta),
    \label{equationapproximategibbsstate}
\end{align}
where $\{\ket{\psi}_j\}$ denote the computational basis. Denote the estimated eigenvalues by $\widehat{\bm\lambda}$, where $\widehat{\lambda}_{j}=\bra{\psi_{j}}U^{\dagger}(\bm\theta)H(\bm\nu)U(\bm\theta)\ket{\psi_j}$ for all $j=1,\ldots,N$. Then, $\rho_{\beta}^*$ is an approximation of $\rho_{\beta}(\bm\nu)$ in the sense that
\begin{align}
    D(\rho_{\beta}^*,\rho_{\beta}(\bm\nu))\leq\sqrt{2\beta\max\left\{\mathbf{E}_{\widehat{\mathbf{p}}^*}[|\widehat{\lambda}-\lambda|],\mathbf{E}_{\mathbf{p}^*}[|\widehat{\lambda}-\lambda|]\right\}}.
\end{align}
where $D(\cdot,\cdot)$ denotes the trace distance, $\bm\lambda$ represent $H(\bm\nu)$'s true eigenvalues, $\mathbf{p}^*$ is the distribution corresponding to $\bm\lambda$, i.e., $\lambda_j=e^{-\beta\lambda_j}/\sum_{l}e^{-\beta\lambda_l}$, and
\begin{align}
    \mathbf{E}_{\widehat{\mathbf{p}}^*}[|\widehat{\lambda}-\lambda|]=\sum_{j=1}^{N}\widehat{p}_j^*|\widehat{\lambda}_j-\lambda_j|,\quad\mathbf{E}_{\mathbf{p}^*}[|\widehat{\lambda}-\lambda|]=\sum_{j=1}^{N}p_j^*|\widehat{\lambda}_j-\lambda_j|.
\end{align}
\end{proposition}

Note that the quantity in Eq.~\eqref{approx_derivative} contains an expectation of distribution $\widehat{\mathbf{p}}^*$, then the partial derivative $\frac{\partial L(\bm\nu)}{\partial\nu_{\ell}}$ is estimated by the sample mean. Specifically, we first randomly select the computational basis vectors $\ket{\psi_j}$ complying with distribution $\widehat{\mathbf{p}}^*$ and then compute the associated eigenvalues via $U(\bm\theta)$. The detailed procedure of sampling and estimate computation is laid out in Algorithm~\ref{alg:grad}. The number of required samples is analyzed in Proposition~\ref{the_gradient}.   
\begin{proposition}[Sample complexity]
\label{the_gradient}
Given $\epsilon>0$ and $\eta\in(0,1)$, Algorithm~\ref{alg:grad} can compute an estimate for the gradient $\nabla L(\bm\nu)$ up to precision $\epsilon$ with probability larger than $1-\eta$. Particularly, the overall number of samples is $KD=O(\beta^2\log(2m/\eta)/\epsilon^2)$ with $K=O(\beta^2/\epsilon^2)$ and $D=O(\log(2m/\eta))$. Besides, the total number of measurements is $O( KD\cdot m\beta^2(n+\log(m/\eta))/\epsilon^2)$.
\end{proposition}
The proofs for Propositions \ref{theoremforgibbsstatepreparation}-\ref{the_gradient} are deferred to the supplementary file.

To validate the gradient estimation, we show that the average of the overall errors determines the accuracy of the gradient estimation. For this purpose, Proposition~\ref{theoremforgibbsstatepreparation} shows that matrix $\rho_{\beta}^*$ is an approximation of the desired density matrix $\rho_{\beta}(\bm\nu)$. Specifically, the trance distance between $\rho_{\beta}^*$ and $\rho_{\beta}(\bm\nu)$ is dependent on the averaged errors $\mathbf{E}_{\widehat{\mathbf{p}}^*}[|\widehat{\lambda}-\lambda|]$ and $\mathbf{E}_{\mathbf{p}^*}[|\widehat{\lambda}-\lambda|]$. Here, notation $|\widehat{\lambda}-\lambda|$ denotes the difference between estimated eigenvalue and the associated real eigenvalue. $\widehat{\mathbf{p}}^*$ and $\mathbf{p}^*$ are probability distributions, corresponding to $\widehat{\bm\lambda}$ and $\bm\lambda$, respectively. In particular, it implies that learning several low-lying eigenvectors with high accuracy may lead to a high-precision estimate of the gradient. We numerically verify this feature in Sec.~\ref{sec:ising}.

Moreover, Proposition~\ref{the_gradient} shows the feasibility of our approach as the number of measurements scales polynomially in parameters $n$, $1/\epsilon$, and $\beta$.

\begin{figure}[!t]
\begin{algorithm}[H]
\footnotesize
\caption{Hybrid quantum-classical Hamiltonian learning algorithm (HQHL)}
\label{alg:hyhl}
\begin{algorithmic}[1]
    \REQUIRE Pauli operators $\{E_{\ell}\}_{\ell=1}^{m}$, constants $\{e_{\ell}\}_{\ell=1}^{m}$, and $\beta$;
    \ENSURE  An estimate for target coefficients $\bm\nu$;
\STATE Initialize coefficients $\{\nu_{\ell}\}_{\ell=1}^{m}$;
\STATE Set number of iterations $I$ and $l=1$;
\STATE Set parameterized quantum circuit $U(\bm\theta)$;
\STATE Set learning rate $r$;
\WHILE{$l\leq I$}
\STATE Set Hamiltonian $H(\bm\nu)=\sum_{\ell=1}^{m}\nu_{\ell}E_{\ell}$;
\STATE Train $U(\bm\theta)$ by SVQE with $H(\bm\nu)$;
\STATE Derive a probability $\widehat{\mathbf{p}}^*$ by performing log-partition function estimation with $U(\bm\theta)$ and $\beta$;
\STATE Compute gradient $\nabla L(\bm\nu)$ by gradient estimation with $U(\bm\theta)$, $\widehat{\mathbf{p}}^*$, and $\beta$;
\STATE Update coefficients $\bm\nu\leftarrow \bm\nu-r\nabla L(\bm\nu)$;
\STATE Set $l\leftarrow l+1$;
\ENDWHILE
\RETURN{the final coefficients $\bm\nu$.}
\end{algorithmic}
\end{algorithm}
\end{figure}
\subsection{Hamiltonian learning algorithm} \label{sec_hla}
Eventually, we present our hybrid quantum-classical algorithm for Hamiltonian learning (HQHL) in Algorithm~\ref{alg:hyhl}. The main idea of HQHL is to find the target interaction coefficients by a gradient-descent method (cf. Figure~\ref{fig:qahl}). Thus, HQHL's main process is to compute the gradient of the objective function. Specifically, we take Pauli operators $\{E_{\ell}\}_{\ell=1}^{m}$, $\{e_{\ell}\}_{\ell=1}^{m}$, and $\beta$ as input. Then we initialize the coefficients by choosing $\bm\nu$ from $[-1,1]^{m}$ uniformly at random. Next, we compute the gradient of the objective function $L(\bm\nu)$ by Algorithm~\ref{alg:grad}. Then update the coefficients by choosing a suitable learning rate $r$ and using the estimated gradient. In consequence, after repeating the training process sufficiently many times, the final coefficients are supposed to approximate the target coefficients $\bm\nu$.

Notably, the learning process is in the ``\textbf{while}" loop of HQHL. In the loop, the subroutine SVQE (cf. Sec.~\ref{sec_hd}) is first called to learn Hamiltonian's eigenvectors and eigenvalues. Here, we choose a suitable parameterized quantum circuit $U(\bm\theta)$ and train it to prepare the eigenvectors of the Hamiltonian $H(\bm\nu)$. Afterwards, we enter the process of the log-partition function estimation (cf. Sec.~\ref{sec_lpfe}). It first exploits the $U(\bm\theta)$ to output the estimated eigenvalues of the parameterized Hamiltonian $H(\bm\nu)$ and then computes the objective function $L(\bm\nu)$. We would obtain a probability distribution $\widehat{\mathbf{p}}^*$ that consists of eigenvalues of the associated Gibbs state $\rho_{\beta}(\bm\nu)= e^{-\beta H(\bm\nu)}/Z_{\beta}(\bm\nu)$. Lastly, we exploit the resultant results (post-training circuit $U(\bm\theta)$ and distribution $\widehat{\mathbf{p}}^*$) to compute the gradients following the procedure in Algorithm~\ref{alg:grad} and update the coefficient $\bm\nu$ accordingly (cf. Eq.~\eqref{eq:coefficient_update}). 

\section{Numerical Results} \label{sec_numerical_sim}
In this section, we conduct numerical experiments to verify the correctness of our algorithm. Specifically, we consider recovering interactions coefficients of several Hamiltonians, including randomly generated Hamiltonians and many-body Hamiltonians. To ensure the performance of the algorithm, we choose a PQC (shown in Fig.~\ref{fig:ansatz_SVQE}) and set the circuit with enough expressibility. When testing our algorithm, we first use SVQE to learn the full spectrum of Hamiltonians, where size of the Hamiltonian varies from 3 to 5. In SVQE, weights $\mathbf{q}$ consists of a normalized sequence of arithmetic sequence. For instance, when $n=3$, $\mathbf{q}=(1,2,3,\ldots,8)/S_3$, where $S_3=\sum_{l=1}^{8}l$. Furthermore, {in order to reduce quantum resources,} we also partially learn the few smallest eigenvalues of the selected Ising models and derive estimates for coefficients up to precision 0.05. With fewer eigenvalues to be learned, the depth of the used PQC is significantly reduced. 
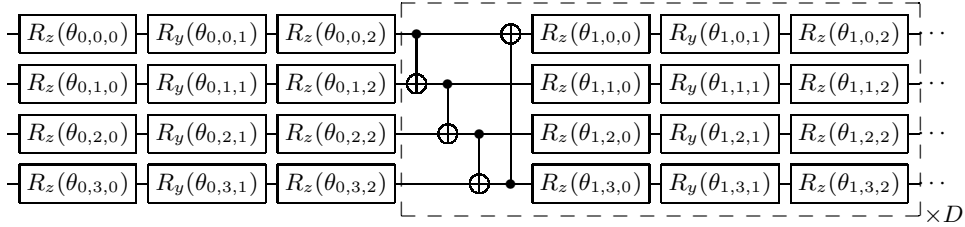
\begin{figure}[hbt]
    \[\Qcircuit @C=0.5em @R=0.5em {
   & \gate{R_z(\theta_{0,0,0})} & \gate{R_y(\theta_{0,0,1})} & \gate{R_z(\theta_{0,0,2})} &\ctrl{+1}&\qw&\qw& \targ & \gate{R_z(\theta_{1,0,0})} & \gate{R_y(\theta_{1,0,1})} & \gate{R_z(\theta_{1,0,2})}&\qw &\cdots \\ 
     & \gate{R_z(\theta_{0,1,0})} & \gate{R_y(\theta_{0,1,1})} & \gate{R_z(\theta_{0,1,2})} &\targ & \ctrl{+1}&\qw& \qw & \gate{R_z(\theta_{1,1,0})} & \gate{R_y(\theta_{1,1,1})} & \gate{R_z(\theta_{1,1,2})}&\qw &\cdots \\ 
& \gate{R_z(\theta_{0,2,0})} & \gate{R_y(\theta_{0,2,1})} & \gate{R_z(\theta_{0,2,2})} &\qw&\targ&\ctrl{+1}& \qw & \gate{R_z(\theta_{1,2,0})} & \gate{R_y(\theta_{1,2,1})} & \gate{R_z(\theta_{1,2,2})}&\qw &\cdots \\ 
& \gate{R_z(\theta_{0,3,0})} & \gate{R_y(\theta_{0,3,1})} & \gate{R_z(\theta_{0,3,2})} &\qw&\qw& \targ &\ctrl{-3}& \gate{R_z(\theta_{1,3,0})} & \gate{R_y(\theta_{1,3,1})} & \gate{R_z(\theta_{1,3,2})}&\qw &\cdots 
  \gategroup{1}{5}{4}{11}{1em}{--} \\
  &&&&&&& &&& &&& \times D 
}\]
    \caption{The selected quantum circuit $U(\bm\theta)$ for stochastic variational quantum eigensolver (SVQE). Here, $D$ represents circuit depth. Parameters $\bm\theta$ are randomly initialized from  a uniform distribution in [0, 2$\pi$] and updated via gradient descent method.}
    \label{fig:ansatz_SVQE}
\end{figure}

\subsection{Random Hamiltonian models}\label{sec:random}
This section shows the effectiveness of our algorithm with random Hamiltonians from three aspects: different $\beta$, different numbers of $\mu$ (\# $\mu$) and a different number of qubits (\# qubits).

\begin{table*}[t]
\footnotesize
  \tabcolsep 11pt
\begin{tabular}{l|c|c|c|c|c|c}
\toprule
Three aspects & \# qubits & \# $\mu$ & $\beta$ & LR & $\mu$ & $E_l$  \\ \hline
 & 3 & 3 & 1  & 1.0  & [ 0.3408 -0.6384 -0.4988]  & [[0 2 1]
 [2 1 3]
 [0 3 3]] \\ \hline
\multirow{2}{*}{Vary $\beta$} & \multirow{2}{*}{3} & \multirow{2}{*}{3} & 0.3  & 8.0 & [-0.4966 -0.8575 -0.7902] & [[1 0 0]
 [3 0 2]
 [3 1 3]]   \\\cline{4-7}
 &  &  & 3   & 0.1 & [0.5717 -0.1313 0.2053] & [[1 0 0]
 [3 3 3]
 [0 2 3]] \\ \hline
\multirow{4}{*}{Vary \# $\mu$} & \multirow{4}{*}{3} & 4 & \multirow{4}{*}{1}  & \multirow{4}{*}{1.0} & [-0.7205 -0.3676 -0.7583 -0.3002] &  [[3 2 1]
 [2 1 3]
 [0 0 2]
 [2 0 0]] \\ \cline{3-3} \cline{6-7}
 &  & 5 &  & & [-0.5254 -0.1481 -0.0037 -0.4373 0.7326]   & [[1 3 0]
 [2 1 1]
 [3 3 2]
 [2 3 1]
 [0 2 0]] \\\cline{3-3}\cline{6-7}
 &  & \multirow{2}{*}{6} &  &  & [-0.5992 0.7912 0.5307 & [[3 2 2]
 [0 2 1]
 [1 2 1]
 \\ 
  &  &  &  &  & -0.5422 -0.9239 0.0354] &
 [2 2 0]
 [0 1 2]
 [3 2 1]] \\ \hline
\multirow{2}{*}{Vary \# qubits} & 4 & \multirow{2}{*}{3} & \multirow{2}{*}{1} & \multirow{2}{*}{1.0}  & [ 0.0858  0.3748 -0.1007]  & [[0 2 0 1]
 [1 0 0 1]
 [2 0 1 0]] \\ \cline{2-2} \cline{6-7}
 & 5 &  &  &  & [-0.0411   0.7882  0.6207] & [[2 2 2 1 2]
 [2 3 3 2 1]
 [1 2 0 2 3]] \\
 \toprule
\end{tabular}
\caption{Hyper-parameters setting. The number of qubits (\# qubits) varies from 3 to 5, and the number of $\mu$ (\# $\mu$) from 3 to 6. $\beta$ is chosen as 0.3, 1, 3. ``LR'' denotes learning rate. The values of $\mu$ are sampled uniformly in the range of [-1, 1]. The term, likes ``[[0 2 1]
 [2 1 3]
 [0 3 3]]'', indicates there are three $E_l$'s and each has three qubits with the corresponding Pauli tensor product. Here ``0,1,2,3'' represent ``$I,X,Y,Z$'' respectively. For example, for the first sample, the corresponding Hamiltonian is taken as $H$=0.3408 $\cdot I\otimes Y \otimes X$ -0.6384 $\cdot Y\otimes X \otimes Z$ -0.4988 $\cdot I\otimes Z \otimes Z$.}
\label{table:param_setting}
\end{table*}

In the experimental setting, we randomly choose Pauli tensor products $E_\ell$ from $\{X,Y,Z,I\}^{\otimes n}$ and target coefficients $\bm\mu$ by a uniform distribution over $[-1,1]$. Specifically, we first vary the values of $\beta$ by fixing the number of $\mu$ and the number of qubits to explore our method's sensitivity to temperature. We similarly vary the number of $\mu$ and the number of qubits by fixing other hyper-parameters to explore our method's scalability. The actual values of these hyper-parameters sampled/chosen in each trial are concluded in Table~\ref{table:param_setting}. In addition, the deep, $D$, of the PQC $U(\bm\theta)$ is set according to the size of Hamiltonian. As number of qubits ranges from $n=3$ to $n=5$, the depth $D$ is set to be $10, 20, 40$, respectively.

In Table~\ref{table:param_setting}, Hamioltonian is represented by a tuple. Each number $0,1,2,3$ corresponds to matrices $I,X,Y,Z$, respectively. $\bm\mu$ denotes the interaction coefficients to be learned. For instance, [[0 2 1] [2 1 3] [0 3 3]] means that the Hamiltonian consists of three Pauli operators, where each term represents a Pauli operator, e.g., [0 2 1] means $I\otimes Y \otimes X$. Then, the parameters in the top second row represents the following Hamiltonian.
\begin{align}
    0.3408 I\otimes Y \otimes X -0.6384 Y\otimes X \otimes Z -0.4988 I\otimes Z\otimes Z.
\end{align}
Other Hamiltonians to be tested are represented in a similar fashion.

The results for these three aspects are illustrated in Fig.~\ref{fig_num}. We find that all curves converge to the values close to 0 in less than ten iterations, which shows our method is effective. In particular, our method works for different $\beta$ means that it is robust to temperature. And the results for the different number of $\mu$ and qubits reveals our method's scalability to a certain extent.

\begin{figure*}[htb]
	\centering
	\subfigure[Vary $\beta$]{
	\includegraphics[width=0.3\textwidth]{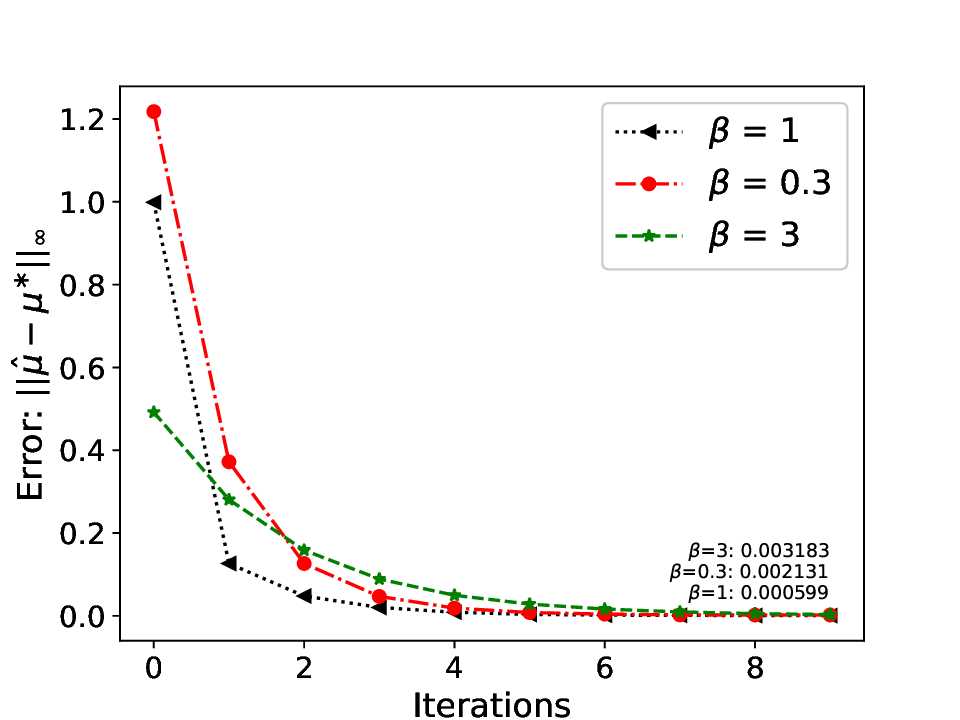}}
\subfigure[Vary \# $\mu$]{
		\includegraphics[width=0.3\textwidth]{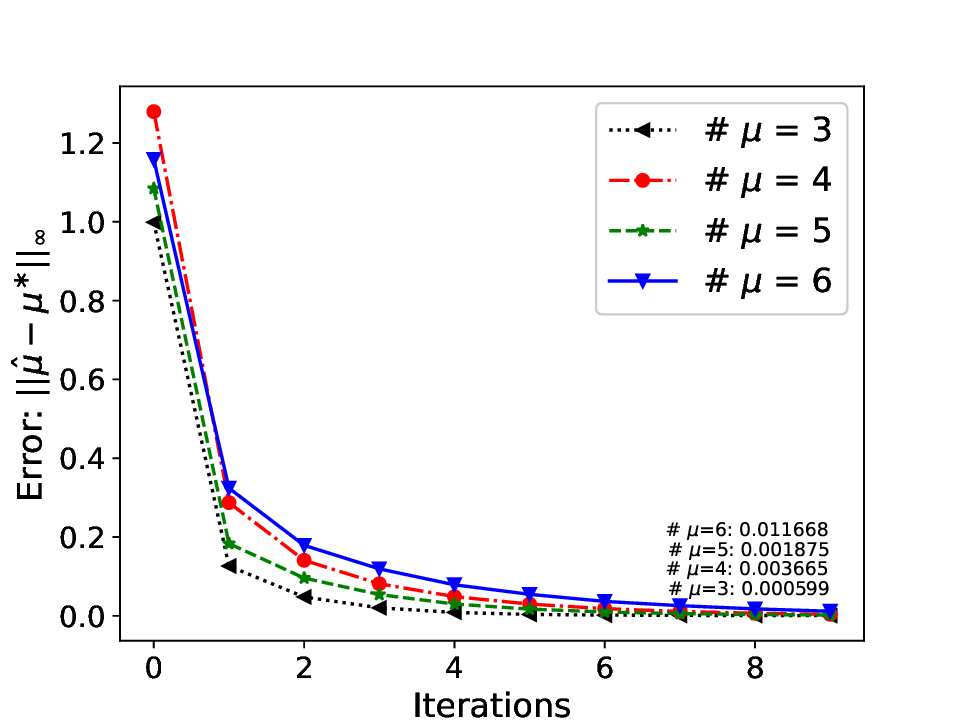}}
	\subfigure[Vary \# qubits]{
		\includegraphics[width=0.3\textwidth]{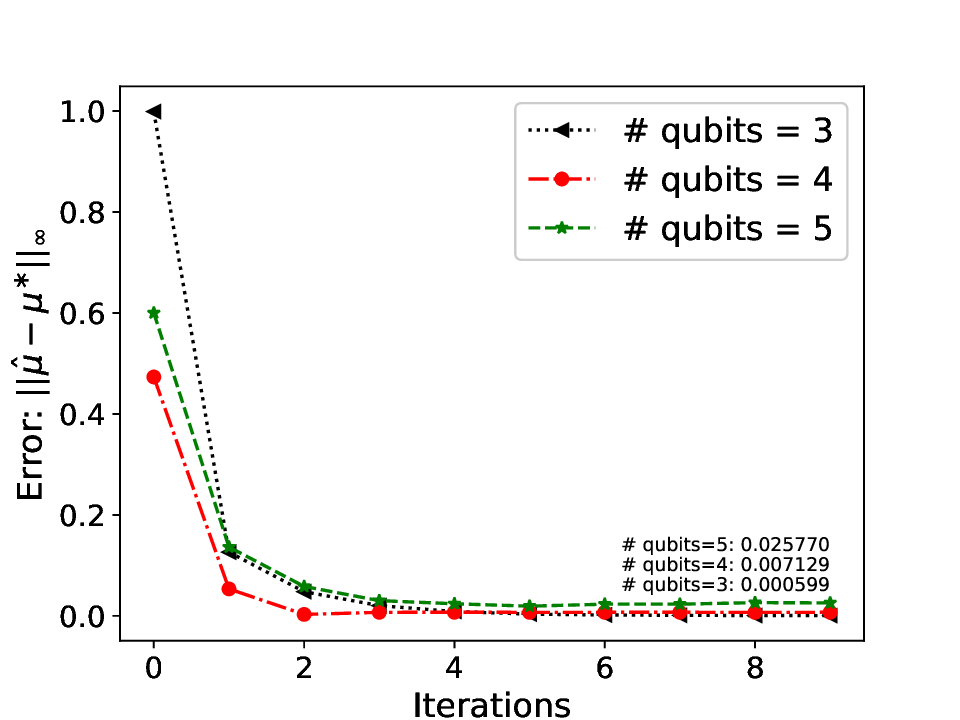}}
		\subfigure[Ising model]{
	\includegraphics[width=0.3\textwidth]{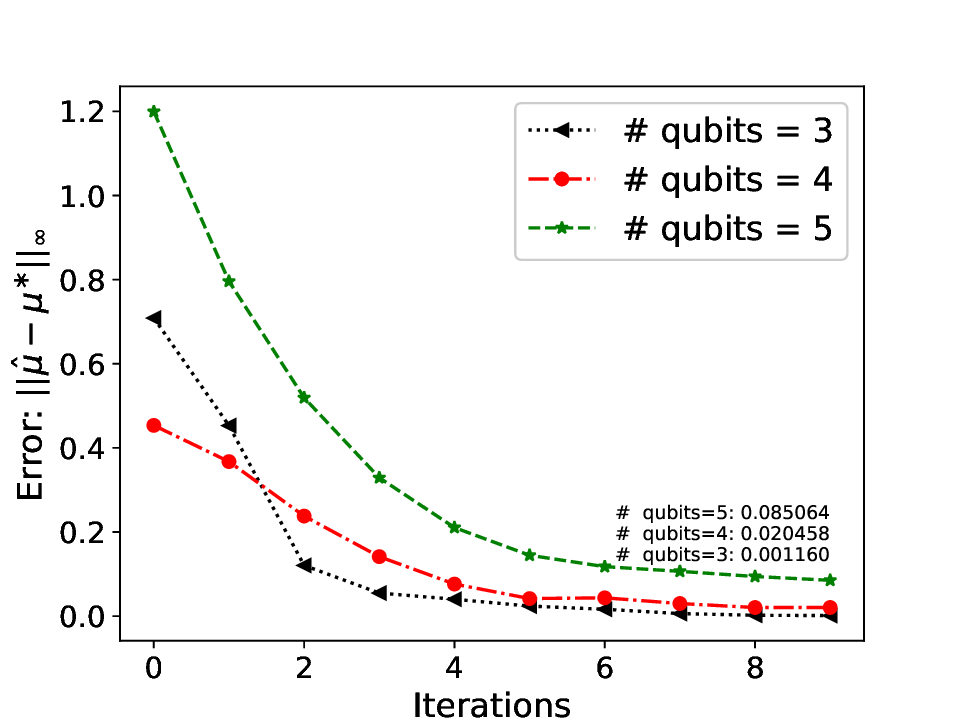}}
\subfigure[XY model]{
		\includegraphics[width=0.3\textwidth]{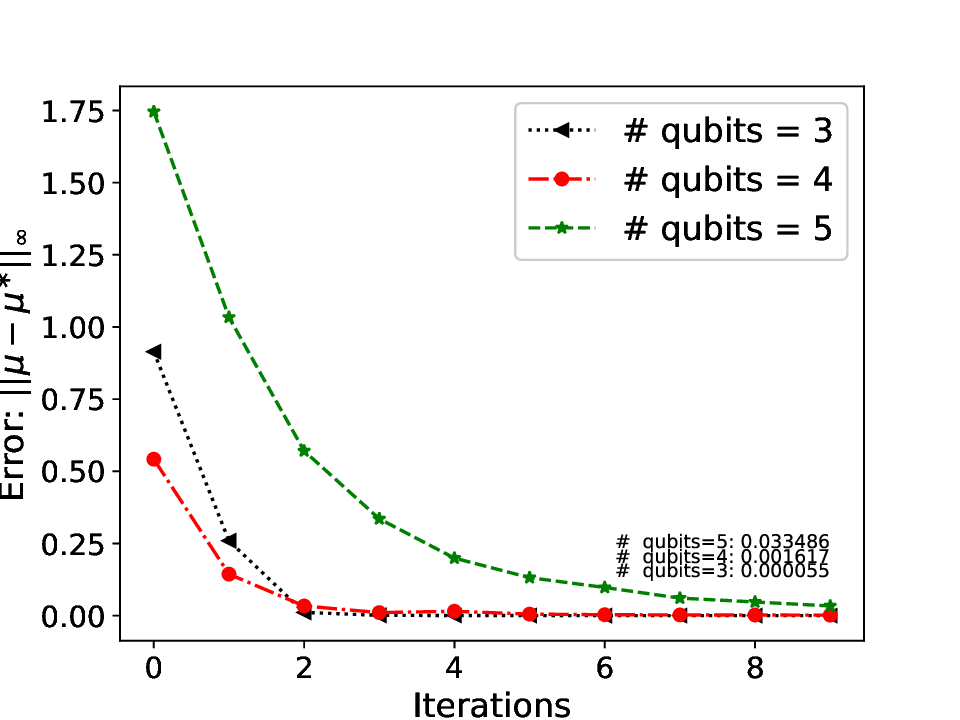}}
	\subfigure[Heisenberg model]{
		\includegraphics[width=0.3\textwidth]{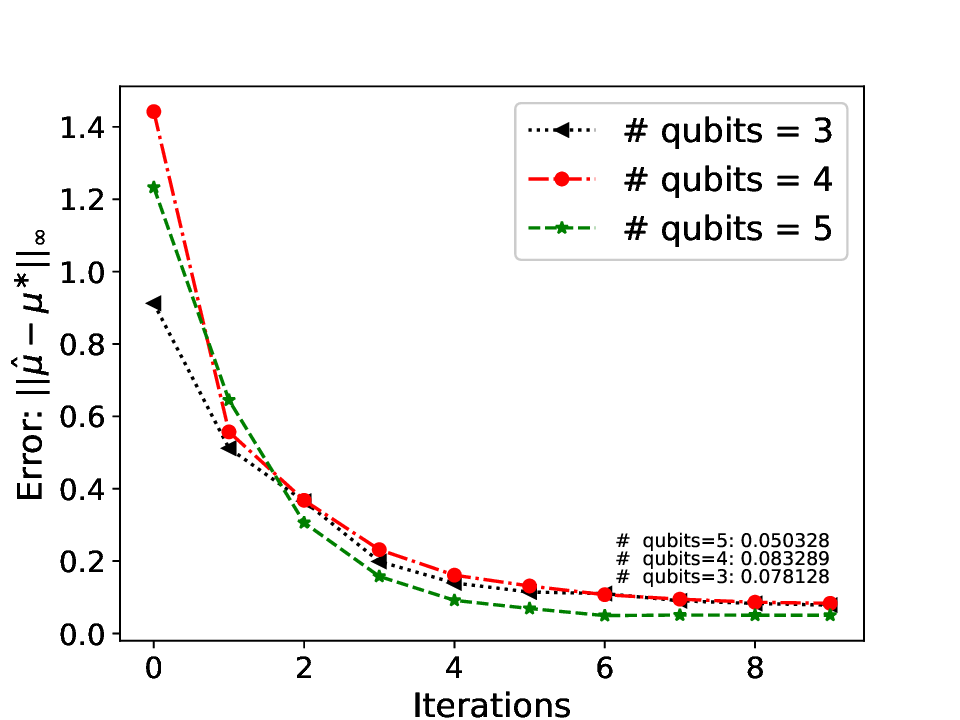}}
\caption{The curves in (a), (b), (c) represent the infinity norm of the error of $\mu$ with different $\beta$, different number of $\mu$, and different number of qubits, respectively. In (d), (e), (f), the curves represent the infinity norm of the error of $\mu$ for different many-body Hamiltonians with the number of qubits varies from 3 to 5. The numbers on the line represent the values of the last iteration. These numbers close to 0 indicate that our algorithm is effective.}
\label{fig_num}
\end{figure*}

\subsection{Quantum many-body models}\label{sec:many_body}
Here, we demonstrate the performance of our algorithm for quantum many-body models. Specifically, we consider the one-dimensional nearest-neighbor Ising model, XY model, and Heisenberg model. These many-body models are described by the Hamiltonians shown below:
\begin{align}
    &(\text{Ising model}) 
    \quad H_{0}=J_{0}\sum_{l=1}^{n} Z^{l}Z^{l+1}+h_{0}\sum_{l=1}^{n}X^{l},\\
    &(\text{XY model})  
    \quad H_{1}=J_{1}\sum_{l=1}^{n}( X^{l}X^{l+1}+ Y^{l}Y^{l+1}),\\
    &(\text{Heisenberg model})  
    \quad H_{2}=J_{2}\sum_{l=1}^{n}(X^{l}X^{l+1}+Y^{l}Y^{l+1}+Z^{l}Z^{l+1})+h_{2}\sum_{l=1}^{n}Z^{l},
\end{align}
where periodic boundary conditions are assumed (i.e., $X^{n+1}=X^{1}$, $Y^{n+1}=Y^{1}$, and $Z^{n+1}=Z^{1}$). Coefficient $J$ is the coupling constant for the nearest neighbor interaction, and $h$ represents the external transverse magnetic field. The experimental parameters are concluded in Table~\ref{tab:many_body_models}. 

We consider the models with a different number of qubits, varying from $n=3$ to $n=5$. The inverse temperature is set as $\beta=1$. The coefficients $J_{0},J_{1},J_{2}$ and $h_{0},h_{2}$ are sampled uniformly from a uniform distribution on [-1,1].  We also employ the parameterized quantum circuit $U(\bm\theta)$ in Fig.~\ref{fig:ansatz_SVQE} for the SVQE. And the depth of $U(\bm\theta)$ is also set as $D=10,20,40$ for different $n$. Moreover, the numerical results are shown in Fig.~\ref{fig_num}, which imply our method is applicable to recover quantum many-body Hamiltonians.
\begin{table}[t]
    \centering
      \tabcolsep 11pt
    \begin{tabular}{l|c|c|c|c|c}
    \toprule
      Many-body & \# qubits & \# $\mu$ & $\beta$ & LR & $\bm\mu$ \\
     models &&&&&  \\
        \hline
      \multirow{3}{*}{Ising model} & 3 & 6 & \multirow{3}{*}{1.0}  & 2.0 & $[J_0=0.1981, \quad h_0=0.7544]$\\\cline{2-3} \cline{5-6}
        %  & &  & 2.0 & $1.0000,1.0000,1.0000]$ \\
        %  \hline
       &  4 & 8 &  & 1.0 & $[J_0=0.5296, \quad h_0= 0.4996]$ \\\cline{2-3} \cline{5-6}
        %  \hline
        & 5 & 10 & & 0.5 & $[J_0=-0.6916,\quad h_0=0.4801]$\\\cline{1-6}
        \hline
        \multirow{3}{*}{XY model} & 3 & 6 & \multirow{3}{*}{1.0} & 1.0 & $J_1=-0.0839$ \\\cline{2-3} \cline{5-6}
         & 4 & 8 &  & 1.0 & $J_1=0.2883$ \\\cline{2-3} \cline{5-6}
         & 5 & 10 & & 0.6 & $J_1=-0.7773$ \\\cline{2-3} \cline{5-6}
         \hline
        \multirow{3}{*}{Heisenberg} & 3 & 12 & \multirow{3}{*}{1.0} & 1.0 & $[J_2= 0.0346,\quad h_2=0.8939]$ \\\cline{2-3} \cline{5-6}
         & 4 & 16 & & 1.0 & $[ J_2=-0.5831,\quad h_2= -0.0366]$\\\cline{2-3} \cline{5-6}
         & 5 & 20 & & 1.0 & $[J_2=0.2883,\quad h_2= -0.2385]$ \\
        \toprule
    \end{tabular}
    \caption{Hyper-parameters setting for many-body models. For each Hamiltonian model, the number of qubits varies from 3 to 5, and the number of $\mu$ is determined by the number of Pauli operators. ``LR" denotes learning rate. The values of $\mu$ are sampled uniformly in the range of $[-1,1]$.}
    \label{tab:many_body_models}
\end{table}

\subsection{Numerical results using fewer eigenvalues of Ising Hamiltonians}\label{sec:ising}
Notice that we use a PQC $U(\bm\theta)$ with deep depths to learn the full spectrum of small-sized Hamiltonians in Secs.~\ref{sec:random}-\ref{sec:many_body}, which may be beyond the capacity of NISQ devices. However, this section demonstrates the efficacy of HQHL in learning the Ising Hamiltonians using a circuit with reduced depth, where few eigenvalues (instead of the full spectrum) are learned. In particular, only halved circuit depths are needed for Hamiltonians with 3-5 qubits, given in Table \ref{tab:many_body_models}. Furthermore, the performance on $n=6$ and $n=7$-qubit Ising models, given below, is tested as well. The presented results imply the potential efficacy of our approach for larger Hamiltonians.
\begin{align}
    H=0.1981\sum_{l=1}^{n} Z^{l}Z^{l+1}+0.7544\sum_{l=1}^{n}X^{l}.
\end{align}

\begin{table}[htb]
\footnotesize
    \tabcolsep 11pt
    \centering
    \begin{tabular}{c|c|c|c|c|c}
    \toprule
        \multirow{2}{*}{\# qubits $n$} & \multirow{2}{*}{weights $\mathbf{q}$} & \multirow{2}{*}{\# $\mu$} & \multirow{2}{*}{depth $D$}  & \multirow{2}{*}{LR} & \multirow{2}{*}{$\#\lambda$}\\
         &  &  &  & & \\
        \hline
        3 & $(0.1,0.2,0.3,0.4,0,\ldots)_{8}$ & 6 & $5$ & 0.4 & 4 \\
        \hline
        4 & $(0.1, 0.15, 0.2, 0.25, 0.3, 0,\ldots)_{16}$ & 8 & $10$ & 0.55 & 5\\
        \hline
        5 & $(0.1, 0.15, 0.2, 0.25, 0.3, 0,\ldots)_{32}$ & 10 & $20$ & 0.7 & 5 \\
        \hline
        6 & $(1/21,2/21,3/21,4/21,5/21,6/21,0,\ldots)_{64}$ & 12 & 30 & 0.55 & 6\\
        \hline
        7 & $(1/21,2/21,3/21,4/21,5/21,6/21,0,\ldots)_{128}$ & 14 & 40 & 0.6 & 6\\
         \toprule
    \end{tabular}
    \caption{Parameters setting for HQHL. The script index means the length of the tuple, e.g., $()_{8}$ indicates the tuple consists of $8$ entries. The notation $0,\ldots$ means the entries following $0$ are all zeros as well. Notation $\#\lambda$ means the number of eigenvalues we learned. Please note that we omit the $\beta=1$ in the table.}
    \label{tab:reduced_circuit}
\end{table}

To reduce the number of eigenvalues to be learned, we tune the weights $\mathbf{q}$ of the SVQE such that the $U(\bm\theta)$ can output several smallest eigenvalues. For instance, five eigenvalues are learned for $4$ \& $5$-qubit Ising Hamiltonians, and four eigenvalues are learned for $3$-qubit Ising Hamiltonians. As a result, the circuit depth of the used $U(\bm\theta)$ is significantly reduced. For example, we only use depth $D=20$ to learn the coefficients with precision $0.05$ for $5$-qubit Ising models. While, in Sec.~\ref{sec:many_body}, we use the depth $D=40$. Moreover, we find out that using a circuit with 35 depths suffices to learn well the 6-qubit Ising model, where SVQE only learns six eigenvalues. Using the circuit with 40 depths could also reach a precision of 0.05 for the 7-qubit Ising Hamiltonian. The details of parameters setting (weights, depth, learning rate, etc.) are given in Table \ref{tab:reduced_circuit}. Besides, the experimental results are depicted in Figure~\ref{fig:reduced_Ising_model}.

\begin{figure}[htb]
    \centering
    \includegraphics[width=0.4\textwidth]{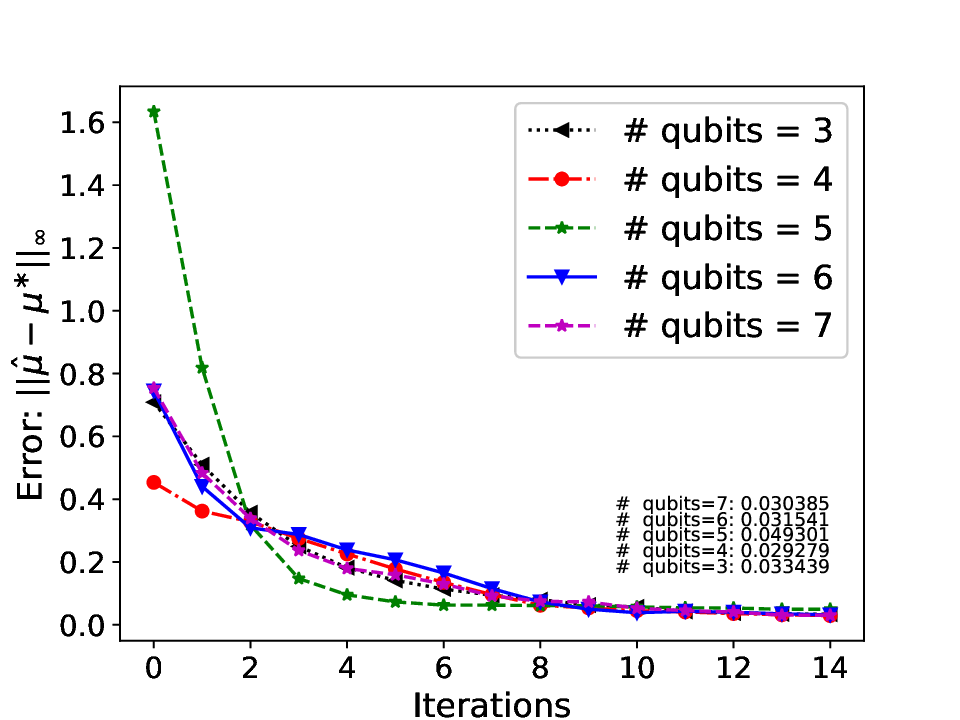}
    \caption{Experimental results by using fewer eigenvalues. Each line corresponds to the results by running HQHL with Ising Hamiltonians of different sizes. Results show that using halved circuit depth, compared to the setting in Sec.~\ref{sec:many_body}, could learn coefficients up to precision $0.05$ for different sized Ising models and a different number of $\mu$.}
    \label{fig:reduced_Ising_model}
\end{figure}

\section{Conclusion}\label{sec_conclusion}
We have proposed a hybrid quantum-classical Hamiltonian learning algorithm that employs a gradient-descent method to find the desired interaction coefficients. We achieve this purpose by unifying the variational quantum algorithms (VQAs) with the strategy proposed in \cite{Anshu2020}. To this end, we develop several subroutines: log-partition function estimation, stochastic variational quantum eigensolver (SVQE), and gradient estimation. In SVQE, we propose a method to learn the full/partial spectrum of the Hamiltonian and use the importance sampling to circumvent the resources in the loss evaluation. In the log-partition function, we propose a method that combines the parameterized quantum circuits and convex optimization to find the global minimum of the free energy as well as compute the log-partition function. In gradient estimation, we present a procedure to compute the gradient of the objective function costing polynomially many resources. Finally, we conduct numerical experiments to demonstrate the effectiveness of our approach with randomly generated Hamiltonians and selected many-body Hamiltonians. In consequence, we show that learning the full spectrum of Hamiltonians in the learning process could produce high-precision estimates of the desired interaction coefficients. Moreover, we also show that partially learning several smallest eigenvalues of Ising Hamiltonians could derive estimates up to a precision of 0.05. Overall, this paper develops a concrete near-term quantum algorithm for Hamiltonian learning and demonstrates the effectiveness as well, which has potential applications in quantum device certification, quantum simulation, and quantum machine learning.

We believe our approach would shed lights on near-term quantum applications. For example, SVQE might enrich the VQE family in the fields of molecules and materials. Moreover, as many problems in computer science can be framed as partition function problems (e.g., counting coloring), our method may contribute to these fields as well. Furthermore, it is reasonable to explore our algorithm's applications in quantum machine learning~\cite{Shingu2020}, quantum error correction~\cite{Valenti2019}, and tomography \cite{Kieferova2017}.

\section*{Acknowledgements}
Y. W. and G. L. acknowledge support from the Australian Research Council (Grant No: DP180100691) and the Baidu-UTS AI Meets Quantum project.  G. L. acknowledges the financial support from China Scholarship Council (No. 201806070139). 

\bibliographystyle{apsrev4-1}
\bibliography{Reference}

\appendix
% \newpage
% \input{appendix}

% \newpage

\onecolumngrid
\begin{center}
{\textbf{\Large Supplementary Material}}
\end{center}

\renewcommand{\theequation}{S\arabic{equation}}
\renewcommand{\thetheorem}{S\arabic{theorem}}
\renewcommand{\thelemma}{S\arabic{lemma}}
\renewcommand{\theproposition}{\arabic{proposition}}
\setcounter{equation}{0}
\setcounter{figure}{0}
\setcounter{table}{0}

% \section{Supplementary proofs}\label{sec_appendix}
\section{Proofs for Eqs.~\eqref{eq_relation_pre}-\eqref{eq_objective_pre}}\label{sec_appendix_par}
Consider a Hamiltonian $H\in\mathbb{C}^{N\times N}$ and a constant $\beta>0$, then the system's free energy is given by $F(\rho)=\tr(H\rho)-\beta^{-1}S(\rho)$. Recall the fact~\cite{nielsen2002quantum} that
\begin{equation}
    S(\rho)\leq-\sum_{j=1}^{N}\rho_{jj}\log\rho_{jj},
\end{equation}
where $\rho_{jj}$ are the diagonal elements of quantum state $\rho$. Using this fact, for any state $\rho$, we can find a lower bound on free energy in the sense that
\begin{align}
    F(\rho)\geq\tr(H\rho)+\beta^{-1}\sum_{j=1}^{N}\rho_{jj}\log\rho_{jj}.\label{free_energy_relation}
\end{align}
On the other hand, let $U$ be a unitary such that $H=U\Lambda U^{\dagger}$, where $\Lambda={\rm diag}(\lambda_{1},...,\lambda_{N})$ is a diagonal matrix. Let $\widetilde{\rho}={\rm diag}(\rho_{11},...,\rho_{NN})$ be the diagonal matrix consisting of $\rho$'s diagonal elements and let $\sigma=U^{\dagger}\widetilde{\rho}U$. It is easy to verify that $\tr(H\rho)=\tr(\Lambda \sigma)$. Furthermore, taking this relation into Eq.~\eqref{free_energy_relation}'s right hand side, we can find that
\begin{align}
    F(\rho)\geq \tr(\Lambda \sigma)-\beta^{-1}S(\sigma).\label{eq_fe_inequality}
\end{align}
Notice that Eq.~\eqref{eq_fe_inequality}'s right-hand side is equal to $F(\widetilde{\rho})$, then we have
\begin{align}
    F(\rho)\geq F(\widetilde{\rho}).\label{eq_fe_inequality_new}
\end{align}
The inequality in Eq.~\eqref{eq_fe_inequality_new} shows that free energy's global optimum is commuting with the Hamiltonian $H$. 

According to the above discussion, we can rewrite the optimization program of finding free energy's minimal value as follows
\begin{align}
    \min_{\rho} F(\rho) = \min_{\mathbf{p}}\sum_{j=1}^{N}\lambda_{j}p_{j}+\beta^{-1}\sum_{j=1}^{N}p_{j}\log p_{j}, \label{pro_optimization}
\end{align}
where $\mathbf{p}$ represents an arbitrary probability distribution. Eq.~\eqref{pro_optimization}'s right-hand side can be solved using the Lagrange multiplier method, and the optimum is given below:
\begin{align}
    \mathbf{p}^*\coloneqq\frac{1}{Z}(e^{-\beta\lambda_{1}},...,e^{-\beta\lambda_{N}}),
\end{align}
with $Z:=\sum_{j=1}^{N}e^{-\beta\lambda_{j}}$.

Finally, the equalities in Eqs.~\eqref{eq_relation_pre}-\eqref{eq_objective_pre} can be proved by taking $\mathbf{p}^*$ into Eq.~\eqref{pro_optimization}'s right-hand side and computing the minimal value.

\section{Proof for Proposition~\ref{theorem_logpartitionfunction_pre}}\label{sec_theorem_3}
\begin{lemma}\label{le_spectral}
For any parameterized Hamiltonian $H(\bm\nu)=\sum_{\ell=1}^{m}\nu_{\ell}E_{\ell}$ with $E_{\ell}\in\{X,Y,Z,I\}^{\otimes n}$, we have
\begin{align}
    \parallel H(\bm\nu)\parallel\leq \sqrt{m}\cdot\parallel\bm\nu\parallel_2.
\end{align}
where $\|\cdot\|$ denotes the spectral norm and $\|\cdot\|_2$ is the $\ell_2$-norm.
\end{lemma}
\begin{proof}
Let $U$ be the unitary that diagonalizes the Hamiltonian $H(\bm\nu)$, and then we can use the following form to represent $H(\bm\nu)$.
\begin{align}
    H(\bm\nu)=\sum_{j=1}^{N}\lambda_{j}\cdot U\op{\psi_j}{\psi_j}U^\dagger,
\end{align}
where $\ket{\psi_1},...,\ket{\psi_N}$ are the computational basis. 

Typically, each eigenvalue is represented as follows:
\begin{align}
    \lambda_{j}&=\bra{\psi_j}U^{\dagger}H(\bm\nu)U\ket{\psi_j}\\
    &=\sum_{\ell=1}^{m}\nu_{\ell} \bra{\psi_j}U^{\dagger}E_{\ell}U\ket{\psi_j}
\end{align}
Then, applying the Cauchy-Schwarz inequality leads to an upper bound on each eigenvalue:
\begin{align}
    (\lambda_{j})^2\leq \sum_{\ell=1}^{m}(\nu_{\ell})^2 \cdot\sum_{\ell=1}^{m}(\bra{\psi_j}U^{\dagger}E_{\ell}U\ket{\psi_j})^2.
\end{align}
Meanwhile, recalling that all $E_{\ell}$ are Pauli matrix tensor product, we can obtain an upper bound below:
\begin{align}
    (\lambda_{j})^2\leq m\sum_{\ell=1}^{m}(\nu_{\ell})^2. \label{eq_cs}
\end{align}
Ranging $j$ in $\{1,...,N\}$ in Eq.~\eqref{eq_cs}, the maximal eigenvalue is upper bounded by $\sqrt{m}\|\bm\nu\|_2$, validating the claim.
\end{proof}

\renewcommand\theproposition{\ref{theorem_logpartitionfunction_pre}}
\setcounter{proposition}{\arabic{proposition}-1}
\begin{proposition}
For any parameterized Hamiltonian $H(\bm\nu)=\sum_{\ell=1}^{m}\nu_{\ell}E_{\ell}$ with $E_{\ell}\in\{X, Y, Z, I\}^{\otimes n}$ and $\bm\nu\in\mathbb{R}^{m}$ and constant $\beta>0$, suppose we are given access to a parameterized quantum circuit $U(\bm\theta)$ that can learn $H(\bm\nu)$'s eigenvectors, then the objective function $C(\mathbf{p})$ can be computed up to precision $\epsilon$ with probability larger than $2/3$ by taking $T=O({m\|\bm\nu\|_2^2}/{\epsilon^2})$ samples. Furthermore, the probability can be improved to $1-\eta$ costing an additional multiplicative factor of $O(\log(1/\eta))$.
\end{proposition}

\begin{proof}
Since the expression $\sum_{j=1}^{N}p_{j}\lambda_{j}$ is regarded as an expectation, then we can estimate it by the sample mean with high accuracy and probability. To be specific, let $X$ denote a random variable that takes value $\lambda_{j}$ with probability $p_{j}$. Then, this expression can be written as
\begin{align}
    \mathbf{E}[X]=\sum_{j=1}^{N}p_{j}\lambda_{j}.
\end{align}
Furthermore, recall Chebyshev's inequality, then we have
\begin{align}
    \Pr\left(|\bar{X}-\mathbf{E}[X]|\leq\epsilon\right)\geq 1-\frac{\mathbf{Var}[X]}{T\epsilon^2}.
\end{align}
where $\bar{X}=\frac{1}{T}(X_1+X_2+...+X_T)$ and $\mathbf{Var}[X]$ is the variance of $X$. Technically, we can set large $T$ to increase the probability. Here, we only need to choose $T$ such that 
\begin{align}
    \frac{\mathbf{Var}[X]}{T\epsilon^2}=\frac{2}{3}.
\end{align}

Note that the second moment $\mathbf{E}[X^2]$ bounds the variance $\mathbf{Var}[X]$. Meanwhile, the second moment of $X$ is bounded by the squared spectral norm of $H$, shown below.
\begin{align}
\mathbf{E}[X^2]&=\sum_{j=1}^{N}p_{j}(\lambda_j)^2\\
&\leq \sum_{j=1}^{N}p_{j}\|H(\bm\nu)\|^2\\
&=\|H(\bm\nu)\|^2.
\end{align}
The inequality is due to the fact that each eigenvalue is less than the spectral norm. Apply Lemma~\ref{le_spectral}, then we will obtain an bound on $T$:
\begin{align}
    T=\frac{3\mathbf{Var}[X]}{2\epsilon^2}\leq \frac{3\mathbf{E}[X^2]}{2\epsilon^2}\leq \frac{3m\|\bm\nu\|_2^2}{2\epsilon^2}.
\end{align}

Lastly, according to the Chernoff bound, we can boost the probability to $1-\eta$ for any $\eta>0$ by repeatedly computing the sample mean $O(\log(1/\eta))$ times and taking the median of all sample means.
\end{proof}

\section{Proof for Proposition~\ref{theorem_SVQE_pre}}
\label{sec_theorem_5}
\begin{lemma}\label{le_ham_mean}
Consider a parameterized Hamiltonian $H(\bm\nu)=\sum_{\ell=1}^{m}\nu_{\ell}E_{\ell}$ with $E_{\ell}\in\{X,Y,Z,I\}^{\otimes n}$. For any unitary $U$ and state $\ket{\psi}$, estimating the value $\bra{\psi}U^{\dagger}H(\bm\nu)U\ket{\psi}$ up to precision $\epsilon$ with probability at least $1-\eta$ requires a sample complexity of
\begin{align}
    O\left(\frac{m\|\bm\nu\|_1^2\log(m/\eta)}{\epsilon^2}\right).
\end{align}
\end{lemma}
\begin{proof}
First, we rewrite the value $\bra{\psi}U^{\dagger}H(\bm\nu)U\ket{\psi}$ as follows:
\begin{align}
    \bra{\psi}U^{\dagger}H(\bm\nu)U\ket{\psi}=\sum_{\ell=1}^{m}\nu_{\ell}\bra{\psi}U^{\dagger}E_\ell{}U\ket{\psi}.
\end{align}

Second, we count the required number of measurements to estimate the value $\bra{\psi}U^{\dagger}E_\ell{}U\ket{\psi}$ up to precision $\epsilon/\|\bm\nu\|_{1}$ with probability at least $1-\eta/m$, where $\|\cdot\|_1$ denotes the $\ell_1$-norm. Since the Pauli operator, $E_{\ell}$, has eigenvalues $\pm1$, we can partition $E_{\ell}$'s eigenvectors into two sets, corresponding to positive and negative eigenvalues, respectively. For convenience, we call the measurement outcome corresponding to eigenvalue 1 as the positive measurement outcome and the rest as the negative measurement outcome. We define a random variable $X$ in the sense that 
\begin{align}
    X=\left\{
    \begin{array}{cc}
    1, & \Pr{[\text{Positive measurement outcome}]}\\
    -1, & \Pr{[\text{Negative measurement outcome}]}
    \end{array}
    \right.
\end{align}
It is easy to verify that $\mathbf{E}[X]=\bra{\psi}U^{\dagger}E_{\ell}U\ket{\psi}$. Thus, an approach to compute value $\bra{\psi}U^{\dagger}E_{\ell}U\ket{\psi}$ is computing an estimate for the expectation $\mathbf{E}[X]$. Meanwhile, consider that $\mathbf{E}[X^2]\leq 1$, then the required number of samples is $O(\|\bm\nu\|_1^2\log(m/\eta)/\epsilon^2)$.

Lastly, for $\bra{\psi}U^{\dagger}H(\bm\nu)U\ket{\psi}$, the estimate's maximal error is $\|\bm\nu\|_1\cdot\epsilon/\|\bm\nu\|_1=\epsilon$. By union bound, the overall failure probability is less than $m\cdot \eta/m=\eta$. Thus, the claim is proved.
\end{proof}

\renewcommand\theproposition{\ref{theorem_SVQE_pre}}
\setcounter{proposition}{\arabic{proposition}-1}
\begin{proposition}
Consider a parameterized Hamiltonian $H(\bm\nu)=\sum_{\ell=1}^{m}\nu_{\ell}E_{\ell}$ with Pauli operators $E_{\ell}\in\{X,Y,Z,I\}^{\otimes n}$ and constants $\nu_{\ell}\in[-1,1]$. Given any constants $\epsilon>0$, $\eta\in(0,1)$, $\beta>0$, the objective function $M(\bm\theta)$ in SVQE can be estimated up to precision $\epsilon$ with probability at least $1-\eta$, costing $TD$ samples with $T=O(m\|\bm\nu\|_2^2/\epsilon^2)$ and $D=O(\log(1/\eta))$. Besides, the total number of measurements is given below:
\begin{align}
    O\left(\frac{mTD\|\bm\nu\|_1^2(n+\log(m/\eta))}{\epsilon^2}\right).
\end{align}
\end{proposition}
\begin{proof}
Let $Y$ denote a random variable that takes value $\bra{\psi_j}U^{\dagger}(\bm\theta)H(\bm\nu)U(\bm\theta)\ket{\psi_j}$ with probability $q_{j}$, then the objective function $M(\bm\theta)$ can be rewritten as 
\begin{align}
    \mathbf{E}[Y]=M(\bm\theta).
\end{align}
By Chebyshev's inequality, the expectation can be computed by taking enough samples of $Y$ and averaging them. Note that the variance of $Y$ determines the number of samples, and the absolute value $Y$ is less than the spectral norm $\|H(\bm\nu)\|$, i.e., $|Y|\leq \|H(\bm\nu)\|$. Along with Lemma~\ref{le_spectral}, it is easy to see that the required number of $Y$'s samples for obtaining an estimate with error $\epsilon/2$ and probability larger than $2/3$ is $T=O(m\|\bm\nu\|_2^2/\epsilon^2)$. Furthermore, by Chernoff bounds, the probability can be improved to $1-\eta/2$ at an additional cost of multiplicative factor of $D=O(\log(1/\eta))$.

On the other hand, each sample $Y$'s value has to be determined by performing the measurement. Since $\ket{\psi_j}$ is a computational basis, hence $Y$ can take at most $2^n$ different values. To ensure the probability for estimating $\mathbf{E}[Y]$ larger than $1-\eta$, the probability of each estimate $\bra{\psi_j}U^{\dagger}(\bm\theta)H(\bm\nu)U(\bm\theta)\ket{\psi_j}$ only needs to be at least $1-\eta/2^{n+1}$. By union bound, the overall failure probability is at most $\eta/2+\eta\cdot\frac{TD}{2^{n+1}}<\eta$ (For large Hamiltonians, the number of samples $TD$ can be significantly less than dimension $2^{n}$). Besides, according to Lemma~\ref{le_ham_mean}, $\bra{\psi_j}U^{\dagger}(\bm\theta)H(\bm\nu)U(\bm\theta)\ket{\psi_j}$'s estimate within accuracy $\epsilon/2$ and probability $1-\eta/2^{n+1}$ requires a sample complexity of $O(m\|\bm\nu\|_1^2(n+\log(m/\eta))/\epsilon^2)$. Thus, the overall number of measurements is the product of the number of samples $TD=O(m\|\bm\nu\|_2^2\log(1/\eta)/\epsilon^2)$ and each sample's sample complexity $O(m\|\bm\nu\|_1^2(n+\log(m/\eta))/\epsilon^2)$. In other words, the objective function $M(\bm\theta)$'s estimate within error $\epsilon$ and probability $1-\eta$ requires a sample complexity of 
\begin{align*}
    O\left(TD\cdot\frac{m\|\bm\nu\|_1^2(n+\log(m/\eta))}{\epsilon^2}\right)=O\left(\frac{m^2\|\bm\nu\|_1^2\|\bm\nu\|_2^2\log(1/\eta)(n+\log(m/\eta))}{\epsilon^4}\right).
\end{align*}
\end{proof}

\section{Proof for Proposition~\ref{theoremforgibbsstatepreparation}}\label{sec:app_theoremforgibbsstatepreparation}
\begin{lemma}\label{le_re_eigen}
Let $\widehat{\bm\lambda}=(\widehat{\lambda}_1,...,\widehat{\lambda}_N)$ denote the estimated eigenvalues from SVQE and define a function $G(\mathbf{p})$ as follows:
\begin{align}
    G(\mathbf{p})\coloneqq \sum_{j=1}^{N} p_{j}\widehat{\lambda}_j+\beta^{-1}\sum_{j=1}^N p_{j}\log p_{j}. \label{loss_g}
\end{align}
Let $\widehat{\mathbf{p}}^*$ be the global optimal point of $G(\mathbf{p})$, that is, for any probability distribution $\mathbf{p}$, we have $G(\widehat{\mathbf{p}}^*)\leq G(\mathbf{p})$. Meanwhile, suppose $\mathbf{p}^*$ is the global optimal point of $C(\mathbf{p})$. Then, we have
\begin{align}
    |G(\widehat{\mathbf{p}}^*)-C(\mathbf{p}^*)|\leq \max\left\{\mathbf{E}_{\widehat{\mathbf{p}}^*}[|\widehat{\lambda}-\lambda|],\mathbf{E}_{\mathbf{p}^*}[|\widehat{\lambda}-\lambda|]\right\},\label{le_eq_ob}
    % \parallel \widehat{\bm\lambda}-\bm\lambda\parallel_\infty.
\end{align}
where 
\begin{align}
    &\mathbf{E}_{\widehat{\mathbf{p}}^*}[|\widehat{\lambda}-\lambda|]=\sum_{j=1}^{N}\widehat{p}_j^*|\widehat{\lambda}_j-\lambda_j|,\\
    &\mathbf{E}_{\mathbf{p}^*}[|\widehat{\lambda}-\lambda|]=\sum_{j=1}^{N}p_j^*|\widehat{\lambda}_j-\lambda_j|.
\end{align}
\end{lemma}
\begin{proof}
% Suppose $\mathbf{p}^*$ denote the global optimum of function $C(\mathbf{p})$. Then, we have
Since functions $C(\mathbf{p})$ and $G(\mathbf{p})$ reach their global minimums at points $\mathbf{p}^*$ and $\widehat{\mathbf{p}}^*$ respectively, then we have
\begin{align}
    &C(\widehat{\mathbf{p}}^*)\geq C(\mathbf{p}^*),\\
    &G(\widehat{\mathbf{p}}^*) \leq G({\mathbf{p}}^*).
\end{align}
Besides, we also have another relation:
\begin{align}
    |C(\mathbf{p})-G(\mathbf{p})|=\sum_{j=1}^{N}p_{j}|(\widehat{\lambda}_j-\lambda_j)|,
    % \leq\parallel\widehat{\bm\lambda}-\bm\lambda\parallel_{\infty},
\end{align}
where $\|\cdot\|_{\infty}$ denotes the maximum norm.

Combining the above inequalities, we have the following result:
% Combining these inequalities results in a relation below.
\begin{align}
C(\mathbf{p}^*)\leq C(\widehat{\mathbf{p}}^*)\leq G(\widehat{\mathbf{p}}^*)+\mathbf{E}_{\widehat{\mathbf{p}}^*}[|\widehat{\lambda}-\lambda|]\leq G({\mathbf{p}}^*)+\mathbf{E}_{\widehat{\mathbf{p}}^*}[|\widehat{\lambda}-\lambda|]\leq C(\mathbf{p}^*)+\mathbf{E}_{\widehat{\mathbf{p}}^*}[|\widehat{\lambda}-\lambda|]+\mathbf{E}_{\mathbf{p}^*}[|\widehat{\lambda}-\lambda|].
\end{align}
Then the inequality in Eq.~\eqref{le_eq_ob} is proved. 
\end{proof}
\renewcommand\theproposition{\ref{theoremforgibbsstatepreparation}}
\setcounter{proposition}{\arabic{proposition}-1}
\begin{proposition}[Correctness]
Consider a parameterized Hamiltonian $H(\bm\nu)$ and its Gibbs state $\rho_{\beta}(\bm\nu)$. Suppose the $U(\bm\theta)$ from SVQE (cf. Algorithm~\ref{alg:svqe}) and $\widehat{\mathbf{p}}^*$ from log-partition function estimation procedure (cf. Algorithm~\ref{alg:lpfe}) are optimal. Define a density operator $\rho_{\beta}^*$ as follows:
\begin{align}
    \rho_{\beta}^*\coloneqq\sum_{j=1}^{N}\widehat{p}_{j}^*\cdot U(\bm\theta)\op{\psi_j}{\psi_j}U^{\dagger}(\bm\theta).
    \label{equationapproximategibbsstate}
\end{align}
where $\{\ket{\psi}_j\}$ denote the computational basis. Denote the estimated eigenvalues by $\widehat{\bm\lambda}$, where $\widehat{\lambda}_{j}=\bra{\psi_{j}}U^{\dagger}(\bm\theta)H(\bm\nu)U(\bm\theta)\ket{\psi_j}$. Then, $\rho_{\beta}^*$ is an approximate of $\rho_{\beta}(\bm\nu)$ in the sense that
\begin{align}
    D(\rho_{\beta}^*,\rho_{\beta}(\bm\nu))\leq\sqrt{2\beta\max\left\{\mathbf{E}_{\widehat{\mathbf{p}}^*}[|\widehat{\lambda}-\lambda|],\mathbf{E}_{\mathbf{p}^*}[|\widehat{\lambda}-\lambda|]\right\}}.
    % \sqrt{2\beta\parallel\widehat{\bm\lambda}-\bm\lambda\parallel_\infty}.
\end{align}
where $D(\cdot,\cdot)$ denotes the trace distance, $\bm\lambda$ represent $H(\bm\nu)$'s true eigenvalues.
\end{proposition}

\begin{proof}
Recalling the expressions of $C({\bf p}^*)$ and $G(\widehat{\bf p}^*)$ in Eqs.~\eqref{op_partion}-\eqref{loss_g}, it is easy to verify the following inequalities:
\begin{align}
    &F(\rho_{\beta}(\bm\nu))=C({\bf p}^*),\\
    &F(\rho_{\beta}^*)=G(\widehat{\bf p}^*).
\end{align}
where $F$ denotes the free energy, i.e., $F(\rho)=\tr(H\rho)-\beta^{-1}S(\rho)$. 

Using the result in Lemma~\ref{le_re_eigen}, we will obtain the following inequality.
\begin{align}
    |F(\rho_{\beta}^*)-F(\rho_{\beta}(\bm\nu))|=|G(\widehat{\bf p}^*)-C({\bf p}^*)|\leq\max\left\{\mathbf{E}_{\widehat{\mathbf{p}}^*}[|\widehat{\lambda}-\lambda|],\mathbf{E}_{\mathbf{p}^*}[|\widehat{\lambda}-\lambda|]\right\}.\label{freeenergyvarepsilon}
    % \parallel\widehat{\bm\lambda}-\bm\lambda\parallel_\infty. 
\end{align}
In the meanwhile, a property of the free energy says that
\begin{align}
    F(\rho_{\beta}^*)=F(\rho_{\beta}(\bm\nu))+\beta^{-1}S(\rho_{\beta}^*\|\rho_{\beta}(\bm\nu)).
\end{align}
where $S(\rho_{\beta}^*\|\rho_{\beta}(\bm\nu))$ is the relative entropy. 
% Replacing $\rho$ with $\rho_{\beta}^*$, the relation above will be rewritten in the following
Rewriting the above equation as follows:
\begin{align}
    F(\rho_{\beta}^*)-F(\rho_{\beta}(\bm\nu))=\beta^{-1}S(\rho_{\beta}^*\|\rho_{\beta}(\bm\nu)).\label{freeenergyrelativee3ntropy}
\end{align}
Combining the relations in Eqs.~\eqref{freeenergyvarepsilon} and \eqref{freeenergyrelativee3ntropy}, we obtain the following inequality:
% Along with relations in , we can obtain 
\begin{align}
    S(\rho_{\beta}^*\|\rho_{\beta}(\bm\nu))\leq \beta\max\left\{\mathbf{E}_{\widehat{\mathbf{p}}^*}[|\widehat{\lambda}-\lambda|],\mathbf{E}_{\mathbf{p}^*}[|\widehat{\lambda}-\lambda|]\right\}.
    % \parallel\widehat{\bm\lambda}-\bm\lambda\parallel_\infty.
\end{align}
Lastly, according to Pinsker's inequality, the above inequality immediately leads to a bound on the trace distance between $\rho_{\beta}$ and $\rho_{\beta}^*$ in the sense that
\begin{align}
     D(\rho_{\beta}^*,\rho_{\beta}(\bm\nu))\leq\sqrt{2 S(\rho_{\beta}^*\|\rho_{\beta}(\bm\nu))}\leq\sqrt{2\beta\max\left\{\mathbf{E}_{\widehat{\mathbf{p}}^*}[|\widehat{\lambda}-\lambda|],\mathbf{E}_{\mathbf{p}^*}[|\widehat{\lambda}-\lambda|]\right\}}.
\end{align}
The the claimed is proved.
\end{proof}

\section{Proof for Proposition~\ref{the_gradient}}\label{sec_theorem_6}
\renewcommand\theproposition{\ref{the_gradient}}
\setcounter{proposition}{\arabic{proposition}-1}
\begin{proposition}[Sample complexity]
% \label{the_gradient}
% The procedure for gradient estimation depicted in Fig.~\ref{fig:gradient} is efficient. To be specific, given $\epsilon>0$ and $\eta\in(0,1/6)$ and all $\ell=1,...,m$, it can estimate each partial derivative $\frac{\partial L(\bm\nu)}{\partial \nu_{\ell}}$ up to precision $\epsilon$ with probability greater than $1/2$ at a cost of sample complexity of $O(\beta^4\log(N/\eta)/\epsilon^4)$, typically, the integer $K=O(\beta^2/\epsilon^2)$. Further, the probability can be improved to $1-\eta$ costing an additional factor of $O(\log(1/\eta))$.
% The procedure for gradient estimation depicted in Fig.~\ref{fig:gradient} is efficient. To be specific, given $\epsilon>0$ and $\eta\in(0,1/6)$ and all $\ell=1,...,m$, it can estimate each partial derivative $\frac{\partial L(\bm\nu)}{\partial \nu_{\ell}}$ up to precision $\epsilon$ with probability larger than $1/2$ at a sample complexity of $O(K^2\beta^2(n+\log(1/\eta))/\epsilon^2)$ with $K=O(\beta^2/\epsilon^2)$. Furthermore, the probability can be improved to $1-\eta$ costing an additional factor of $O(\log(1/\eta))$.
% The procedure for gradient estimation depicted in Fig.~\ref{fig:gradient} is efficient. 
Given $\epsilon>0$ and $\eta\in(0,1)$, Algorithm~\ref{alg:grad} can compute an estimate for the gradient $\nabla L(\bm\nu)$ up to precision $\epsilon$ with probability larger than $1-\eta$. Particularly, the overall number of samples is $KD=O(\beta^2\log(m/\eta)/\epsilon^2)$ with $K=O(\beta^2/\epsilon^2)$ and $D=O(\log(2m/\eta))$. Besides, the total number of measurements is $O( KD\cdot m\beta^2(n+\log(m/\eta))/\epsilon^2)$.

% Given $\epsilon>0$ and $\eta\in(0,1)$, Algorithm~\ref{alg:grad} can compute an estimate for the gradient $\nabla L(\bm\nu)$ up to precision $\epsilon$ with probability larger than $1-\eta$. In specific, the overall number of samples is $KD=O(\beta^2\log(m/\eta)/\epsilon^2)$ samples with $K=O(\beta^2/\epsilon^2)$ and $D=O(\log(m/\eta))$. Besides, the total number of measurements is $O(K^2\beta^2(n+\log(1/\eta))/\epsilon^2)$. 

% Given $\epsilon>0$ and $\eta\in(0,1/6)$, Algorithm~\ref{alg:grad} can compute an estimate for the gradient $\nabla L(\bm\nu)$ up to precision $\epsilon$ with probability larger than $1/2$, costing $K=O(\beta^2/\epsilon^2)$ samples and $O(K^2\beta^2(n+\log(1/\eta))/\epsilon^2)$ measurements. Furthermore, the probability can be improved to $1-\eta$ at an expense of a multiplicative factor of $O(\log(1/\eta))$.
% Given $\epsilon>0$ and $\eta\in(0,1/6)$, Algorithm~\ref{alg:grad}  can compute an estimate for each partial derivative $\frac{\partial L(\bm\nu)}{\partial \nu_{\ell}}$ up to precision $\epsilon$ with probability larger than $1/2$ costing a sample complexity of $O(K^2\beta^2(n+\log(1/\eta))/\epsilon^2)$ with $K=O(\beta^2/\epsilon^2)$. Furthermore, the probability can be improved to $1-\eta$ costing an additional factor of $O(\log(1/\eta))$.
\end{proposition}
% \begin{proposition}
% % [Restatement of Theorem~\ref{the_gradient}]
% The procedure for gradient estimation depicted in Fig.~\ref{fig:gradient} is efficient. To be specific, given $\epsilon>0$ and $\eta\in(0,1/6)$ and all $\ell=1,...,m$, it can estimate each partial derivative $\frac{\partial L(\bm\nu)}{\partial \nu_{\ell}}$ up to precision $\epsilon$ with probability larger than $1/2$ at a sample complexity of $O(K^2\beta^2(n+\log(1/\eta))/\epsilon^2)$ with $K=O(\beta^2/\epsilon^2)$. Furthermore, the probability can be improved to $1-\eta$ costing an additional factor of $O(\log(1/\eta))$.
% \end{proposition}
\begin{proof}
% Recall the expression of partial derivative in Eq.~\eqref{approx_derivative}, and then we set the estimate's error as $\epsilon/\beta$ to ensure the gradient's maximal error less than $\epsilon$. 
Let $Z_{\ell}$ denote the random variable that takes value $\bra{\psi_j}U^{\dagger}(\bm\theta)E_{\ell}U(\bm\theta)\ket{\psi_j}$ with probability $\widehat{p}_{j}^*$, for all $\ell=1,...,m$. Then we have
\begin{align}
    \mathbf{E}[Z_{\ell}]=\sum_{j=1}^N \widehat{p}_{j}^*\cdot\bra{\psi_j}U^{\dagger}(\bm\theta)E_{\ell}U(\bm\theta)\ket{\psi_j}.\label{grad_mean}
\end{align}
Thus partial derivative can be computed in the following way
\begin{align}
    \frac{\partial L(\bm\nu)}{\partial \nu_{\ell}}\approx -\beta \mathbf{E}[Z_{\ell}]+\beta e_{\ell}.
\end{align}
It implies that the estimate's error can be set as $\epsilon/\beta$ to ensure the gradient's maximal error less than $\epsilon$.

Next, we determine the number of samples such that the overall failure probability for estimating the gradient is less than $\delta$. Since the gradient has $m$ partial derivatives, corresping to $\mathbf{E}[Z_{\ell}]$, thus it suffices to estimate each with probability larger than $1-\delta/m$. Meanwhile, each mean $\mathbf{E}[Z_{\ell}]$ can be computed by sampling. Notice that all $|Z_{\ell}|\leq1$, by Chebyshev's inequality, then it suffices to take $K=O(\beta^2/\epsilon^2)$ samples to compute an estimate for each $\mathbf{E}[Z_{\ell}]$ with precision $\epsilon/2\beta$ and probability larger than $2/3$. Furthermore, by Chernoff bounds, the probability can be improved to $1-\eta/2m$ at an additional cost of multiplicative factor of $D=O(\log(2m/\eta))$. It is worth pointing out that, for each variable $Z_{\ell}$, the samples are taken according to the same probability distribution $\widehat{\mathbf{p}}^*$, thus it is natural to use the sampled states $\ket{\psi_{t_{j}^s}}$ (cf. Algorithm~\ref{alg:grad}) to compute all means $\mathbf{E}[Z_{\ell}]$. Then the total number of samples is $KD=O(\beta^2\log(m/\eta)/\epsilon^2)$.

% By union bound, it suffices to suppress the failure probability for estimating each $\mathbf{E}[Z_{\ell}]$ to $\delta/m$. On the other hand, , then 
% the required number $K$ of samples for estimating $\mathbf{E}[Z]$ with error $\epsilon/2\beta$ and probability $2/3$ is $K=O(\beta^2/\epsilon^2)$.
% The proof follows closely to that of Proposition~\ref{theorem_SVQE_pre}. 
% Since we aim to estimate the gradient with error $\epsilon$, then each partial derivative has to be estimated up to precision $\epsilon/\beta$. 

On the other hand, each value $\bra{\psi_j}U^{\dagger}(\bm\theta)E_{\ell}U(\bm\theta)\ket{\psi_j}$ in Eq.~\eqref{grad_mean} has to be computed by performing the measurement. Note that there are $2^{n}$ values $\bra{\psi_j}U^{\dagger}(\bm\theta)E_{\ell}U(\bm\theta)\ket{\psi_j}$ in all. To ensure the mean estimate's failure probability less than $\eta/2m$, it suffices to suppress each value's failure probability to $\eta/2^{n+1}m$. Following the same discussion in Lemma~\ref{le_ham_mean}, the estimate for value $\bra{\psi_j}U^{\dagger}(\bm\theta)E_{\ell}U(\bm\theta)\ket{\psi_j}$ can be computed up to precision $\epsilon/2\beta$ using $O(\beta^2\log(2^{n+1}m/\eta)/\epsilon^2)$ measurements. 

% Note that random variable $Z_{\ell}$ can take at most $2^n$ different values, $\bra{\psi_j}U^{\dagger}(\bm\theta)E_{\ell}U(\bm\theta)\ket{\psi_j}$. we estimate each with probability $1-\eta/2^n$ such that the overall failure probability for these $2^n$ values is less than $\eta$. Setting the error as $\epsilon/2\beta$,
% % Furthermore, let estimate's error be $\epsilon/2\beta$. Next, 
% computing such an estimate for each value $\bra{\psi_j}U^{\dagger}(\bm\theta)E_{\ell}U(\bm\theta)\ket{\psi_j}$ with probability larger than $1-\eta/2^n$ requires a sample complexity of $O(\beta^2\log(2^n/\eta)/\epsilon^2)$. 
Regarding the failure probability, by union bound, the overall failure probability is at most $m\cdot(\eta/2m+KD\cdot\eta/2^{n+1}m)$, where $KD$ is the number of samples $KD=O(\beta^2\log(m/\eta)/\epsilon^2)$. Especially, for larger Hamiltonians, the number of measurements is usually less than the dimension $2^{n}$. Thus, the overall failre probability is less than $\eta$.

Lastly, the total number of measurements is given below:
\begin{align}
    m\cdot KD\cdot O(\beta^2\log(2^{n+1}m/\eta)/\epsilon^2)=O(m\beta^4\log(m/\eta)\log(2^{n+1}m/\eta)/\epsilon^4).
\end{align}

\end{proof}
\end{document}